\documentclass[12pt, oneside]{article}   

\usepackage[margin = 1in]{geometry} 
\usepackage{graphicx} 		
\usepackage{mathtools}
\usepackage{color}
\usepackage{hyperref}
\usepackage{enumerate}
\usepackage{stix}
\usepackage[round]{natbib}
\usepackage{authblk}
\usepackage{setspace}
\usepackage{amsfonts}
\usepackage{multirow}
\usepackage{amsthm, amsmath, amssymb}
\usepackage{mathtools}
\usepackage{enumerate}
\usepackage{comment}
\usepackage{longtable}

\newtheorem{theorem}{Theorem}[section]

\allowdisplaybreaks
\hypersetup{
     colorlinks   = true,
     citecolor    = black,
     linkcolor    = blue
}

\newcommand\norm[1]{\left\lVert#1\right\rVert}

\newcommand{\diag}{\rm diag}
\newcommand{\Var}{\rm Var}
\newcommand{\argmax}{\rm argmax}

\let \tilde \widetilde
\let \hat \widehat

\title{Adjusted Blockwise Empirical Likelihood Ratio\\ Confidence Region With Proper Finite Sample \\Coverage Accuracy For Weakly Dependent Data}
\author[a]{Guangxing Wang\thanks{\noindent Corresponding author.\newline \textcolor{white}{****} E-mail address: kenwang@ucdavis.edu.}}
\author[b]{Wolfgang Polonik}
\affil[a]{Division of Public Health Sciences, Fred Hutchinson Cancer Research Center, Seattle, WA 98115 USA}
\affil[b]{Department of Statistics, University of California, Davis, Davis, CA 95616 USA}

\providecommand{\keywords}[1]{\textbf{\textit{Keywords:}} #1}

\begin{document}
\maketitle

\singlespacing
\begin{abstract}
    {\color{black}It is well known that the empirical likelihood ratio confidence region suffers finite sample under-coverage issue, and this severely hampers its application in statistical inferences.} The root cause of this under-coverage is an upper limit imposed by the convex hull of the estimating functions that is used in the construction of the profile empirical likelihood. For i.i.d data, various methods have been proposed to solve this issue by modifying the convex hull, but it is not clear how well these methods perform when the data are no longer independent. In this paper, we propose an adjusted blockwise empirical likelihood that is designed for weakly dependent multivariate data. We show that our method not only preserves the much celebrated asymptotic $\chi^2-$distribution, but also improves the finite sample coverage probability by removing the upper limit imposed by the convex hull. Further, we show that our method is also Bartlett correctable, thus is able to achieve high order asymptotic coverage accuracy. 
\end{abstract}


\section{Introduction}

Empirical likelihood has been studied extensively in the past three decades as a reliable and flexible alternative to the parametric likelihood. Among its numerous attractive properties, the ones that are most celebrated are the asymptotic $ \chi^{2}$-distribution of the empirical likelihood ratio and the ability to use Bartlett correction to improve the corresponding confidence region coverage accuracy. However, despite these desirable properties that parallel the parametric likelihood methods, there is a serious drawback, namely the corresponding confidence region has an under-coverage problem for finite samples of multivariate variables. {\color{black} When this under-coverage happens, usually with a relatively small sample size given a fixed number of parameters, an empirical likelihood ratio confidence region seemingly constructed for a coverage probability of say $95\%$ will in effect only provide coverage with a probability much less than $95\%$. Therefore, if proper care is not given when using empirical likelihood in statistical inferences, one may arrive at misleading and dubious conclusions, especially in this big data age, when one may easily encounter data with hundreds of parameters and complex structures such as data dependency.} This undesirable feature of empirical likelihood was noticed early on, for example by \cite{Owen1988}. {\color{black} Then over a decade later, \cite{Tsao2004} discovered that the root cause of this under-coverage issue was the convex hull constraint used in the construction of the profile empirical likelihood (cf. Equation \ref{eq: overview profile el} in Section \ref{sec: el overview}).} For independent data, various methods have been proposed to address this issue. They can be divided into two main thrusts, (i) improving the approximation to the limiting distribution of the log empirical likelihood ratio, and (ii) tackling the convex hull constraint, the real culprit. As for (i), among others, \cite{Owen1988} proposed to use a bootstrap calibration, and \cite*{DiCiccio1991} showed that by scaling the empirical likelihood ratio with a Bartlett factor, the limiting coverage accuracy can be improved from $ O(n^{-1}) $ to $ O(n^{-2}) $. For (ii), there exist three major methods aiming to solve this convex hull constraint, namely the penalized empirical likelihood by \cite{Bartolucci2007}, the adjusted empirical likelihood by \cite*{Chen2008}, and the extended empirical likelihood by \cite{Tsao2013}. These three methods then have been extended and refined by subsequent research. For instance, \cite{Zhang2016} extended the penalized empirical likelihood to weakly dependent data by using a fixed-b blockwise method. \cite{Emerson2009} proposed a modified adjusted empirical likelihood method. \cite{Liu2010} showed that by choosing the tuning parameter in the adjusted empirical likelihood in a specific way, it is possible to achieve the Bartlett corrected coverage error rate. \cite{Chen2013} studied the finite sample properties of the adjusted empirical likelihood and discussed a generalized version of the method proposed in \cite{Emerson2009}.

It is worth pointing out that most of existing works have focused on independent data, and the aforementioned \cite{Zhang2016} was the first paper to address the convex hull constraint for weakly dependent data with penalized empirical likelihood \citep{Bartolucci2007} under the blockwise framework, which was introduced to empirical likelihood by \cite{Kitamura1997}. \cite*{PiyadiGamage2017} studied the adjusted empirical likelihood for time series models with the frequency domain empirical likelihood (FDEL) and recently \cite*{Chen2020BartlettVariance} studied the Bartlett corrected FDEL with unknown innovation variance. These recent developments on FDEL are no doubt important contributions to empirical likelihood methods for time series data; however, as pointed out by \cite{Nordman2014}, the FDEL is restricted to a class of normalized spectral parameters and assumes the time series has a linear representation in terms of an iid sequence of mean-zero innovations and a linear filter sequence (see the model setups in \citeauthor{PiyadiGamage2017}, \citeyear{PiyadiGamage2017} and \citeauthor{Chen2020BartlettVariance}, \citeyear{Chen2020BartlettVariance} for an example). Therefore, the FDEL is inapplicable to other parameters such as the process mean, parameters in regression settings, and a broad range of other applications. In addition, unlike \cite{PiyadiGamage2017}, \cite{Chen2020BartlettVariance} did not consider the finite sample coverage upper bound issue caused by the convex hull constraint; thus, the finite sample performance of their method relying solely on Bartlett correction is unclear when the parameter dimension is relatively large with respect to a given sample size. In this paper, we propose an adjusted blockwise empirical likelihood (ABEL) for weakly dependent data targeting a broad range of applications beyond the FDEL. We show that  ABEL preserves the much celebrated asymptotic $ \chi^{2}$-distribution and it has a neat theoretical justification from the Bartlett correction point of view. More importantly, ABEL is free from the finite sample convex hull constraint, thus is designed to provide proper coverage probability even when the parameter dimension is relatively large compared to a given sample size. ABEL mirrors the adjusted empirical likelihood in the independent setting; however, the application of ABEL in the weakly dependent setting is far from trivial due to the data blocking used to deal with data dependency. In particular, it requires meticulous calculations (see Section \ref{sec: adj proofs} and Appendix S1) to ensure the correct convergence rates of many different parts in arriving at the coverage error rate that parallels the Bartlett corrected error rate, which is the golden theoretical bench mark.

This paper is organized as follows. Section \ref{sec: el overview} gives a brief introduction to the empirical likelihood method and its convex hull constraint problem. Basic notation used throughout the paper is also established in this section. Section \ref{sec: adj bel} introduces ABEL along with its asymptotic properties. In Section \ref{sec: adj bart tune}, we show that ABEL ratio confidence region can be theoretically justified by a smaller coverage probability error that is on par with the Bartlett corrected error rate. In Section \ref{sec: adj simulation}, we demonstrate the performance of the ABEL method through a simulation study. In section \ref{sec: applicatioin}, we demonstrate the application of ABEL to an economic dataset under the regression setting. Proofs of the theoretical results are relegated to Section \ref{sec: adj proofs}. Detailed calculations for some of the equations in Section \ref{sec: adj proofs} and additional simulation results are given in the online supplement Sections S1 and S2 respectively.

\section{Empirical likelihood and the convex hull constraint} \label{sec: el overview}

We begin with a brief introduction to the empirical likelihood (EL) methods and the convex hull problem that causes the EL ratio confidence region's poor coverage probability with finite sample. For a comprehensive review of the empirical likelihood methodology, we refer to \cite{Owen2001}.  Let $ x_{1}, \dotsc, x_{n} \in \mathbb{R}^{m}$ be i.i.d random samples from an unknown distribution $ F(x) $ and  $ \theta \in \mathbb{R}^{p} $ be the parameter of interest. Let $ g(x; \theta): \mathbb{R}^{m+p} \mapsto \mathbb{R}^{q}$ be a $q$-dimensional estimating function, such that $ \mathbb{E}[g(x;\theta_{0})]  = 0 $, where $ \theta_{0} $ is the true parameter. One of the advantages of the empirical likelihood is that additional information about the parameter can be incorporated through the estimating equations (\citeauthor{Qin1994}, \citeyear{Qin1994}). In other words, we can have $ q \geq p $. The profile empirical likelihood is defined as 
\begin{equation} \label{eq: overview profile el}
EL_{n}(\theta) = \sup_{p_{i}} \left\{\prod_{i = 1}^{n} p_{i}: \sum_{i = 1}^{n} p_{i} = 1, p_{i} \geq 0, \sum_{i = 1}^{n} p_{i} g(x_{i}; \theta) = 0 \right\}.
\end{equation}

If the supremum is taken over the empty set, then by convention, one sets $EL_n(\theta) = -\infty$.  The profile empirical likelihood ratio is defined as $ ELR_{n}(\theta) = n^{n} EL_{n}(\theta).$ Under regularity conditions, see for example \cite{Qin1994}, it can be shown that $-2 \log ELR_{n}(\theta_{0}) \rightarrow_{d} \chi^{2}_{q} \text{ as } n \rightarrow \infty, $ and an asymptotic $ (1 - \alpha)\cdot100\% $ empirical likelihood confidence region for $ \theta $ is given by
$\{\theta: -2 \log ELR_{n}(\theta) < \chi^{2}_{q, 1-\alpha}\}.$ These results are the most celebrated properties of the empirical likelihood, paralleling their parametric counterpart \citep{Wilks1938TheHypotheses}. Despite these advantages, it {\color{black} was} noted by \cite{Owen1988}, that the empirical likelihood confidence region {\color{black}often under-covers for finite samples}. \cite{Tsao2004} studied the least upper bounds on the coverage probabilities by using the fact that $ EL_{n}(\theta) $ is finite if and only if $ 0 $ is in the convex hull $ \mathcal{H} $ of \{$ g(x_{i};\theta), i = 1, \dotsc, n \} $, showing that the empirical likelihood confidence region coverage probability is upper bounded by the probability of 0 being in $ \mathcal{H} $. Furthermore, \cite{Tsao2004} demonstrated that this upper bound is affected by sample size and parameter dimension in such a way that if the parameter dimension is comparable to the sample size, then the upper bound goes to 0 as the sample size goes to infinity. This not only explains the root of the under-coverage issue, but it also shows the severity of the upper bound problem when the sample size is small compared to the parameter dimension. 

{\em Dependent data:} In addition to the convex hull constraint, another difficulty facing the empirical likelihood given by \eqref{eq: overview profile el} is the fact that it fails to self-normalize under dependent data (the implicit covariance structure fails to account for the dependency amount data). As a result, inference based on the asymptotic $\chi^2-$distribution is no longer valid \citep[cf.][]{Kitamura1997}. Two different approaches were proposed to extend empirical likelihood to dependent data. One is the aforementioned FDEL originally introduced by \cite{Monti1997EmpiricalModels} and a more general approach is the blockwise empirical likelihood (BEL) proposed by \cite{Kitamura1997}. Many variants of these two approaches have been proposed since then, with the FDEL variants focusing on the somewhat limited time series setting and the BEL variants mostly focusing on the choices of data blocking. However, as pointed out by \cite{Zhang2016}, the BEL approach in general is not free from the convex hull constraint, thus suffers from the finite sample coverage bound issue. Intuition dictates that the finite sample coverage bound issue will be severe for BEL than for the original EL in the independent setting because the data blocking effectively reduces sample size thus reducing the probability of the convex hull containing $0$. \cite{Zhang2016} proposed to use a penalized approach with the BEL that effectively breaks the convex hull constraint; however, their test statistic has a pivotal but non-traditional asymptotic distribution that requires to be simulated in practice. In addition, it is unclear whether the penalized BEL is able to achieve the Bartlett corrected error rate. In the next section, we introduce the ABEL that is free from the finite sample convex hull constraint, has a test statistic with a conventional and easy to use asymptotic $\chi^2-$distribution, and is able to achieve a confidence region coverage error rate that rivals the Bartlett corrected rate. The newly proposed ABEL is general enough to be applicable under most of the BEL variants with little or no modifications. For ease of conveying the basic idea, we focus on the data blocking scheme originally proposed by \cite{Kitamura1997} when introducing ABEL. In the simulation and data application, we demonstrate how ABEL can work with other BEL variants. In particular, an automatic data blocking scheme proposed by \cite{Kim2013} is applied with the ABEL framework.

\section{Adjusted blockwise empirical likelihood} \label{sec: adj bel}

We begin by introducing some additional notations used in the data blocking, which ensures the proposed ABEL self-normalizes properly. Let $ M, L, \text{ and } Q = \lfloor (n-M)/L \rfloor + 1 $ be the block length, the gap between block starting points, and the number of blocks respectively, where $ M \rightarrow \infty$ as $ n \to \infty$, and $L \leq M$. Define the blockwise estimating equations as $T_{i}(\theta) := 1/M \sum_{k = 1}^{M} g(x_{(i-1)L + k }; \theta),$ which is an average of the estimating equations within each data block. As mentioned earlier, various blocking schemes exist and we are focusing on the most generic one here for ease of introduction. However, it is clear that no matter what blocking scheme is used, the probability of $0$ being in the convex hull of $ \{T_{i}(
\theta), \dotsc, T_{Q}(\theta)\} $ is not increased, and in many cases, it is greatly reduced. In other words, the finite sample upper bound on the likelihood ratio confidence region's coverage probability is not improved. To solve this problem, we construct a pseudo blockwise estimating equation as follows

\begin{equation} \label{eq: adj block extra point}
    T_{Q+1}(\theta) := -a \overline{T}(\theta),
\end{equation}
where $ \overline{T}(\theta) = \frac{1}{Q}  \sum_{i=1}^{Q} T_{i}(\theta)$ and $ a > 0 $. Now it is easy to see that with this extra blockwise estimating equation the convex hull of $ \{T_{1}(\theta), \dotsc, T_{Q}(\theta), T_{Q+1} (\theta)\} $ will always contain $ 0 $, and this effectively eliminates the nontrivial coverage upper bound for the confidence region based on the ABEL ratio defined below. The idea of adding the pseudo blockwise estimating equation is similar to the adjusted empirical likelihood proposed by \cite{Chen2008} for i.i.d data. Unlike in the i.i.d setting, here we have to deal with the extra difficulties imposed by data dependency and data blocking and have to be careful with the interplay between the rate of the tuning parameter $ a $ and the rate of the block length $ M $ (cf. Section \ref{sec: adj proofs} for more details). 

Now, we define the adjusted blockwise empirical likelihood with $ T_{1}, \dotsc, T_{Q}, T_{Q+1} $ as
\begin{equation} \label{eq: adj blockwise el}
    ABEL_{Q}(\theta) = \sup_{p_{i}} \left\{\prod_{i = 1}^{Q+1} p_{i}: p_{i} \geq 0, \sum_{i = 1}^{Q+1} p_{i} =1, \sum_{i = 1}^{Q+1} p_{i} T_{i}(\theta) = 0 \right\},
\end{equation}
and with a standard Lagrange argument, the log adjusted empirical likelihood ratio becomes
\begin{equation} \label{eq: adj blockwise elr}
    ABELR_{Q} (\theta) = \log \frac{ABEL_{Q}(\theta)}{(Q+1)^{-(Q+1)}} = - \sum_{i=1}^{Q+1} \log [1+\lambda_a^{\top}(\theta) T_{i}(\theta)],
\end{equation}
where $\lambda_{a}(\theta) \in \mathbb{R}^q$ is the vector of Lagrange multipliers satisfying $0 = \sum_{i=1}^{Q+1} \{T_{i}(\theta)\}/\{1+\lambda_a^{\top}(\theta)T_{i}(\theta)\}.$ We assume $\sum_{i=1}^{Q+1} T_i(\theta) T_i^\top (\theta)$ to have full rank, then $\lambda_a(\cdot)$ is continuously differentiable. For details on the basic properties of empirical likelihood, we refer to \cite{Owen1990}, \cite{Qin1994}, and \cite{Kitamura1997}. In the rest of this paper, we omit the dependence on $\theta$ and write $\lambda_a$ unless ambiguity arises. 

{\color{black}The addition of $T_{Q+1}(\theta)$ effectively eliminates the convex hull constraint and allows ABEL to be well defined for any $\theta$. In particular, this simple yet effective modification does not require additional conditions compared to BEL by \cite{Kitamura1997}, except a mild assumption on how fast the adjustment parameter $ a $ can increase with sample size in order to maintain the asymptotic properties. For ease of reference, we list the conditions for BEL before we present our theoretical results.} A detailed discussion of these assumptions can be found in \cite{Kitamura1997}. They are generalizations from the assumptions used in the i.i.d setting (\citeauthor{Qin1994}, \citeyear{Qin1994}) to the weakly dependent setting. The main assumptions are on the continuity and differentiability of the estimating function $g(\cdot,\cdot)$ around the true parameter of interest; so that the remainder terms in the Taylor expansion of the log empirical likelihood ratio are controlled, and that the dominating term converges to a $\chi^2-$distribution. We assume that $ x_{i}, i = 1, \dotsc, n $ is a sample from a stationary stochastic process $ \{X_{i}\} $ that satisfies a strong mixing assumption:
\begin{equation} \label{eq: adj strong mixing 1}
    \alpha_{X}(k) \rightarrow 0 \text{ as } k \rightarrow \infty,
\end{equation}
where $ \alpha_{X}(k) = \sup_{A,B} |P(A \cap B) - P(A) P(B)|, A \in \mathcal{F}^{0}_{-\infty}, B \in \mathcal{F}^{\infty}_{k} $, and $ \mathcal{F}^{n}_{m} = \sigma(X_{i}, m \leq i \leq n) $ denotes the $\sigma$-algebra generated by $X_i,\ m \leq i \leq n$. Further, assume that for some $c >1$,
\begin{equation}\label{eq: adj strong mixing 2}
    \sum_{k=1}^{\infty} \alpha_{X}(k) ^{1-1/c} < \infty.
\end{equation}

\begin{enumerate}
	\item[A.1] The parameter space $ \Theta \subset \mathbb{R}^{p} $ is compact.
	\item[A.2] $ \theta_{0} $ is the unique root of $ \mathbb{E} g(x_{i}; \theta) = 0 $.
	\item[A.3] For sufficiently small $ \delta >0 \ \text{ and } \ \eta >0, \mathbb{E} \sup_{\theta^{*} \in \Gamma(\theta, \delta) }  \norm{g(x_{t};\theta^{*})}^{2(1+\eta)} < \infty, \text{ for all } \theta \in \Theta$, where $\Gamma(\theta, \delta)$ is an open ball with midpoint $\theta$ and radius $\delta$.
	\item[A.4] $\lim_{\theta_j \to \theta}g(x, \theta_j) = g(x,\theta)$ is continuous for all $x$ except for a null set, which may vary with $ \theta $.
	\item[A.5]  $ \theta_{0} $ is an interior point of $ \Theta, $ and $ g(x; \theta) $ is twice continuously differentiable at  $ \theta_{0} $.
	\item[A.6] $ \Var(\frac{1}{\sqrt{n}}\sum_{i=1}^{n} g(x_{i}; \theta_{0})) \rightarrow S \in \mathbb{R}^{q \times q}, S >0, \ \text{ as } \ n \rightarrow \infty $. 
	\item[A.7] With $ c >1 $ from (\ref{eq: adj strong mixing 2}), we have $ \mathbb{E} \norm {g(x; \theta_{0})}^{2c} < \infty $. Also, there exists $K < \infty$ with $\mathbb{E} \sup\limits_{\theta^{*} \in \Gamma(\theta_{0}, \delta) }  \norm{g(x;\theta^{*})}^{2+\epsilon} < K;$  
	$\mathbb{E} \sup\limits_{\theta^{*} \in \Gamma(\theta_{0}, \delta) } \norm{\frac{\partial g(x;\theta^{*})}{\partial \theta^{\top}}}^{2}<K $ and 
	$ \mathbb{E} \sup\limits_{\theta^{*} \in \Gamma(\theta_{0}, \delta) } \norm{\frac{\partial^{2} g_{j}(x;\theta^{*})}{\partial \theta \partial \theta^{\top}}}<K, \text{ where } g_{j}(x;\theta) $ is the $j$th component of $ g(x; \theta) $. Moreover, $M \to \infty,$ with $M = o(n^{1/2-1/(2+\epsilon)}),$ for some $\epsilon >0 $.
	\item[A.8]  $ D = \mathbb{E} \frac{\partial g(x; \theta_{0})}{\partial \theta^{\top}} $ is of full rank.
\end{enumerate}

The following theorem then shows that under the above assumptions, the ABEL ratio has an asymptotic $\chi^{2}$-distribution.

\begin{theorem} \label{thm: adj asymptotic chi}
	Suppose that assumptions A.1-A.8 and the strong mixing condition \eqref{eq: adj strong mixing 2} hold. If $ a = o(n/M) $, then with the true parameter $\theta_0$,
	\begin{equation*}
	    -2 \frac{n}{QM}ABELR_{Q}(\theta_{0}) \rightarrow_{d} \chi^{2}_{q}, \text{ as } n \rightarrow \infty.
	\end{equation*}
\end{theorem}

The factor $ n/(QM) $ is to account for the overlap between blocks. For non-overlapping blocks, $n/(QM) = 1 $. {\color{black} The theoretical rate of the tuning parameter $o(n/M)$ is still quite large, even though it is slower than the $o(n)$ under the i.i.d setting in \cite*{Chen2008} due to the data blocking used here for weakly dependent data. The choice of the tuning parameter $ a $ in practice is delicate, and it may depend on the statistical task at hand. In the next section, we show that with a properly selected tuning parameter, the ABEL ratio confidence region defined below in \eqref{eq: adj elr confidence region} can achieve the theoretical Bartlett corrected error rate. Similar phenomenon has been observed under the i.i.d setting by \cite{Liu2010}; however, extensive effort and extreme care must be exercised in our weakly dependent setting to deal with the interplay between data blocking and the adjustment. (cf. see details in Section \ref{sec: adj proofs} and Appendix S1)}

By Theorem \ref{thm: adj asymptotic chi}, a $(1-\alpha)\cdot100\%$ asymptotic confidence region based on the ABELR can be constructed as,
\begin{equation} \label{eq: adj elr confidence region}
    CR_{1-\alpha} = \left\{\theta | -2 \frac{n}{MQ} ABELR_Q(\theta) < \chi^2_{q, 1 -\alpha} \right\}.
\end{equation}

By the design of the extra point (Equation \ref{eq: adj block extra point}), it is clear that $ABELR_Q(\theta)$ is well defined for any $\theta$. As a consequence, there is no finite sample upper bound imposed by the convex hull on the coverage probability of the confidence region \eqref{eq: adj elr confidence region}. 

In practice, one may be interested into the following hypothesis
\begin{align} \label{sec meth: testing constraint}
    H_0: (\theta_{i_1}, \dotsc, \theta_{i_r})^\top = (\theta_{i_1, 0}, \dotsc, \theta_{i_r, 0})^\top
\end{align}
for a subset $(\theta_{i_1}, \dotsc, \theta_{i_r})^\top$ of the parameters $\theta = (\theta_1, \dotsc, \theta_p)^\top$, $r \leq p$. For example in regression analysis, $Y = \theta X + \epsilon, \theta^\top \in \mathbb{R}^p, X \in \mathbb{R}^p$, where it is usually of interest to test the significance of a subset of the coefficients to see if the corresponding covariates have effects. The following result shows that the ABEL can accommodate this situation. {\color{black} Let $\Theta_{r,0}$ denote the parameter space under the null hypothesis \eqref{sec meth: testing constraint}.} Define the maximum blockwise empirical likelihood estimator (MBELE) as
\begin{align} \label{sec meth: mbele}
    \hat{\theta} := \underset{\theta \in \Theta_{r,0}}{\argmax} \sup_{p_{i}} \left\{\prod_{i = 1}^{Q} p_{i}: p_{i} \geq 0, \sum_{i = 1}^{Q} p_{i} =1, \sum_{i = 1}^{Q} p_{i} T_{i}(\theta) = 0 \right\}.
\end{align}

\begin{theorem} \label{sec meth: adj asym chi under constraint}
    Suppose that assumptions A.1-A.8 and the strong mixing condition \eqref{eq: adj strong mixing 2} hold. If $ a = o(n/M) $, then with $\hat{\theta}$ as defined in \eqref{sec meth: mbele},
	\begin{equation*}
	    -2 \frac{n}{QM}ABELR_{Q}(\hat{\theta}) \rightarrow_{d} \chi^{2}_{q-p+r}, \text{ as } n \rightarrow \infty.
	\end{equation*}
\end{theorem}
{\color{black} As is well known from standard likelihood theory, the degrees of freedom $p-q+r$ accounts for the fact that $p-r$ parameters have to be estimated.
}

\section{Tuning Parameter for Bartlett Corrected Error Rate} \label{sec: adj bart tune}

Being Bartlett correctable is an important feature of the parametric likelihood ratio confidence region, where in the i.i.d case, the coverage probability error can be decreased from $ O(n^{-1}) $ to $ O(n^{-2}) $. As its parametric counterpart, \citet*{DiCiccio1991} showed that the empirical likelihood for smooth function model is also Bartlett correctable. Further, \cite{Chen2007} showed that this property also holds for the empirical likelihood with general estimating equations. All these treat the i.i.d case. For weakly dependent data, \cite{Kitamura1997} showed that the BEL for smooth function model is Bartlett correctable, where the coverage probability error can be improved from $ O(n^{-2/3}) $ to $ O(n^{-5/6}) $. In this section, we show that through an Edgeworth expansion of the ABEL ratio, a tuning parameter $ a $ can be found such that with general estimating equations the confidence region \eqref{eq: adj elr confidence region} has coverage error $ O(n^{-5/6}) $, and we call such an $a$ the high precision tuning parameter. For ease of notation, we assume the non-overlapping blocking scheme, where $ M = L $. We assume the block size $ M = O(n^{1/3}) $, and $ QM \geq n $. In addition to the mixing condition \eqref{eq: adj strong mixing 2}, we need a stronger assumption that controls the degrees of dependence $ \alpha_{X}(m) \leq c e^{-dm}, \text{ for all } m, $ where $ \alpha_{X}(m) $ and $c$ are defined in \eqref{eq: adj strong mixing 1} and \eqref{eq: adj strong mixing 2} and $d$ is a positive constant. {\color{black} We further assume the validity of the Edgeworth expansion of sums of dependent data. Following \cite{Gotze1983}, this entails the assumptions of (i) the existence of higher moments, (ii) a conditional Cramer condition, and (iii) the random processes are approximated by other exponentially strong mixing processes with exponentially decaying mixing coefficients that satisfy a Markov type condition.} \cite{Kitamura1997} used an Edgeworth expansion for blockwise empirical likelihood. Edgeworth expansion for sums of data blocks were also used to analyze bootstrap methods by \cite{Lahiri1991SecondBootstrap}, \cite{Davison1993}, and \cite{Lahiri1996OnModels}. 

To simplify notation, assume that 
\begin{align} \label{sec meth: identity covariance}
    V :=M \mathbb{E} [T_i(\theta_0) T_i(\theta_0)^{\top}] = I_q,
\end{align}
where $I_q$ is the identity matrix, and we omit the argument $\theta_0$ from quantities related to $T_i$ in the rest of this paper unless it is instructive to emphases the dependence on $\theta_0$. Let \textcolor{black}{the superscript $ j $ denote the jth component of a random vector, for example} $ T_i^j $ denotes the jth component of $ T_i $. For $j_k \in \{1, \dotsc, q\}, k = 1, \dotsc, v $, define 
\begin{equation} \label{eq: adj moment notation}
    \alpha^{j_{1} \cdots j_{v}} := M^{v-1} \mathbb{E} [T_i^{j_{1}} \cdots T_i^{j_{v}}],
\end{equation}
and notice that $ \alpha^{rr} = 1$ and $ \alpha^{rs} = 0 \text{ for } r \neq s $. 

Further, we define the counterpart of \eqref{eq: adj moment notation} for dependent data as the following: for integers $ 0 < k(1) < \dotsc < k(d) = k,\ k \geq 3$, let
{\small
\begin{align} \label{sec tuning: alpha tilde}
&\tilde{\alpha}^{j_1 \cdots j_{k(1)}, j_{k(1) + 1} \cdots j_{k(2)},    \cdots j_{k(d - 1)}, j_{k(d - 1) + 1} \cdots j_{k(d)}} \\
:= & \frac{1}{Q} \sum_{1 \leq i(1), \dotsc, i(d) \leq Q} 
\mathbb{E} \left\{M^{-1}\left(M^{k(1)} T_{i(1)}^{j_{1}} \cdots T_{i(1)}^{j_{k(1)}}\right) 
\times \left( M^{k(2) - k(1)} T_{i(2)}^{j_{k(1) + 1}} \cdots T_{i(2)}^{j_{k(2)}}\right) \right. \nonumber \\
&\hspace*{1.5cm}\left. \times \cdots \times \left(M^{k(d) - k(d-1)} T_{i(d)}^{j_{k(d - 1) - 1}} \cdots T_{i(d)}^{j_{k(d)}}\right) \right\} 
\times I_{\left\{\max_{p, q < d}|i(p) - i(q)| \leq k-2 \right\}}. \nonumber
\end{align}
}
\textcolor{black}{With the above higher moments $ \alpha^{\dotsc} $ and $ \tilde{\alpha}^{\dotsc} $, the high precision tuning parameter $a$ can be expressed} as follows. Let 
\begin{equation} \label{eq: adj bart tuning}
    a : = \frac{1}{2q} \frac{Q}{n} a_{ii},
\end{equation}
where, for $r,\ i  = 1, \dotsc, q$,
\begin{equation} \label{eq: adj ari}
    a_{ri} = \frac{1}{q} \left(t_{1a} [2] + t_{1b} [2] + t_{1c} + t_{2a} [2] + t_{2b} + t_{3a} [2] + t_{3b} [2] + t_{3c}\right),
\end{equation}
with $ t_{ja} [2] = t_{ja} + t_{ja'},\ j = 1,2,3$, and similarly for $t_{jb}[2]$, where the quantities $t_{ja},\ t_{jb},\ t_{jc}, \ t_{ja'}, \ t_{jb'}$,\ $j = 1,2,3$ are defined as follows. First we give $t_{ja},\ t_{jb},\ t_{jc}$,\ $j = 1,2,3$:
\begin{align*}
&t_{1a} = \alpha^{rkl} \tilde{\alpha}^{i, k, l},\\
&t_{1b} = \frac{3}{8} \tilde{\alpha}^{rk, l} \tilde{\alpha}^{lk, i} - \frac{5}{6} \alpha^{rkl}\tilde{\alpha}^{ik, l} - \frac{5}{6} \alpha^{rkl} \tilde{\alpha}^{kl, i} + \frac{8}{9} \alpha^{rkl} \alpha^{ikl}, \\
&t_{1c} = \frac{1}{4} \alpha^{rkl} \tilde{\alpha}^{il, k} - \frac{2}{3} \alpha^{rkl} \tilde{\alpha}^{ik, l} + \frac{2}{9} \alpha^{rkl} \alpha^{ikl}, \\
&t_{2a} = \frac{3}{8} \tilde{\alpha}^{rl, l} \tilde{\alpha}^{ik, k} - \frac{5}{12} \alpha^{irk} \tilde{\alpha}^{kl, l}  + \frac{4}{9} \alpha^{ril} \alpha^{lkk}  - \frac{5}{12} \alpha^{kll} \tilde{\alpha}^{ik, r},\\
&t_{2b} = \frac{1}{4} \tilde{\alpha}^{rk, k} \tilde{\alpha}^{il,l} -\frac{1}{3} \alpha^{rkk}  \tilde{\alpha}^{il, l}  + \frac{1}{9} \alpha^{rkk} \alpha^{ill}, \\
& t_{3a} = -\frac{1}{2} \tilde{\alpha}^{rk, ki},\\
&t_{3b} = \frac{3}{8} \tilde{\alpha}^{rk, ik} + \tilde{\alpha}^{irl, l} - \frac{3}{4} \alpha^{rikk},\\
&t_{3c} = \frac{1}{4} \tilde{\alpha}^{rk, ik}.
\end{align*}

The $t_{ja'} \text{ and } t_{jb'},\ j = 1,2,3$ are the same as $t_{ja} \text{ and } t_{jb},\ j = 1,2,3$ respectively, except that the superscripts $r \text{ and } i$ are exchanged, for example $t_{1a'} = \alpha^{ikl} \tilde{\alpha}^{r, k, l}$.

The following theorem shows that, when using this tuning parameter $ a $ given by Equation (\ref{eq: adj bart tuning}), the ABEL ratio confidence region \eqref{eq: adj elr confidence region} achieves the Bartlett corrected coverage error rate.

\begin{theorem} \label{thm: high order adj coverage}
	Suppose that conditions A.1-A.8 in Section \ref{sec: adj bel} hold with non-overlapping blocking scheme, and that $\alpha_X(m) \leq c e^{-dm}$ for some $c, d> 0$. We further assume that the assumptions for the Edgeworth expansion for sums of dependent data mentioned in the beginning of this section hold. If $ a $ is defined as in \eqref{eq: adj bart tuning}, then as $ n \rightarrow \infty$, 
	\[
	P\left(- 2 \frac{n}{MQ} \sum_{i = 1}^{Q+1} \log \{1 + \lambda_{a}^{\top} T_i(\theta_0) \} \leq x \right) = P(\chi^{2}_{q} \leq x) + O(n^{-5/6}).
	\]
\end{theorem}

{\color{black}
The high precision tuning parameter $a$ is composed of various orders of population moments, which can be replaced by their corresponding sample moments to obtain a plug-in estimator $\hat{a}_{ii}$ of $ a_{ii} $. However, the plug-in estimator for the independent counterpart of $a$ given by \cite{Liu2010} is observed to be severely biased. We expect that the plug-in estimator for $a$ will be biased because with weakly dependent data, the high precision tuning parameter is more complex and involves even more higher moments than its independent counterpart. Under the independent setting, \cite{Liu2010} discussed the bootstrapping idea from Bartlett correction given by \cite{Chen2007} since the high precision tuning parameter is closely related to the Bartlett correction factor. However, to the best of our knowledge, there is no existing discussion on practical estimation procedure for the Bartlett correction factor for the blockwise empirical likelihood with weakly dependent data. In particular, the bootstrap procedure in \cite{Chen2007} will not work for our high precision tuning parameter because to implement their bootstrap procedure with our ABEL, it is required to compute the empirical likelihood ratio that requires the high precision tuning parameter, which is the target that we want to estimate. In other words, this procedure requires the exact object that we want to estimate in the intermediate steps. One possible solution is to use an iterative algorithm and hope it will convergence. The difficulties facing such an approach are that 1) there is no known such algorithm, and 2) an iterative algorithm is likely to increase the already heavy computational burden associated with bootstrap. Since the main drawback of the plug-in estimator is its large bias, we propose to correct this bias via a blockwise bootstrap procedure \citep[cf. see Chapter 10 in][for details on bootstrap bias correction procedure]{Efron1993AnBootstrap}. The blockwise bootstrap is to deal with dependent data and a comprehensive account of this method is given in \cite{Lahiri2003ResamplingData}. Various blocking schemes such as moving block bootstrap and nonoverlappping block bootstrap (NBB) can be used depending on the application. In the simulation and application given in Sections \ref{sec: adj simulation} and \ref{sec: applicatioin} respectively, we use the NBB for demonstration purposes. As pointed out by \cite{Efron1993AnBootstrap}, the bootstrap bias correction may introduce undue standard errors. As a remedy, we compare the estimated bias and standard error before carrying out bias correction. If the estimated bias is smaller than the estimated standard error, we use the original plug-in estimator, if the estimated bias is larger than the estimated standard error, we proceed to correct the bias on the plug-in estimator. In addition, one can also estimate the standard error of the bias corrected plug-in estimator and compare it to the standard error of the original plug-in estimator to check if the bias correction procedure introduces excessive standard error.
}

In addition, the estimated high precision tuning parameter may not always be positive. If it is, then the convex hull constructed with the extra point will always contain the origin. If, on the other hand, $ \hat{a}_{ii} $ is negative, then $ \hat{a} = \frac{1}{2} \frac{Q}{n} \frac{1}{q} \hat{a}_{ii} $ is also negative. As a result, the convex hull with the new point added will not contain the origin if the original convex hull does not. To avoid the second situation, if $ \hat{a}_{ii} < 0 $, we adopt the method proposed by \cite{Liu2010} to add two extra points $ T_{n+1} = -a_{1} \overline{T} \text{ and } T_{n+2} = -a_{2} \overline{T}$, such that $ a = a_{1} + a_{2} $. We can let $ a_{1} = 2a = \frac{Q}{n} \frac{1}{q} a_{ii} < 0 $ and $ a_{2} = -a = -\frac{1}{2} \frac{Q}{n} \frac{1}{q} a_{ii}> 0 $, such that $ T_{n+2} $ will guarantee that the origin is in the new convex hull. Moreover, since $ a = a_{1} + a_{2} $, adding $ T_{n+1} \text{ and } T_{n+2} $ will have the same effect as adding $ T_{n+1} $ with tuning parameter $a$ in terms of obtaining the Bartlett coverage probability. 

\section{Simulation} \label{sec: adj simulation}

In this section, we examine the numerical properties of the proposed ABEL approach through a simulation study. We compare the corresponding confidence regions based on ABEL under various tuning parameters to the one based on BEL. The data $x_i,\ i = 1, \dotsc, n$ are simulated from the following AR(1) model 
\begin{equation*}
    x_{i+1} = \text{diag}(\rho) x_i + \epsilon_{i+1},~~i = 1, \dotsc, n,
\end{equation*}
where $\epsilon_i$ are i.i.d $d-$dimensional multivariate standard normal random variables and $\diag(\rho)$ is a diagonal matrix with $\rho$ on the diagonal. The parameter of interest is the population mean $\mu := E(x_i)$. In order to see how the data dependencies affect the performance of the methods, we simulate the data with a broad range of $\rho's$. In particular, we look at $ \rho = -0.8, -0.5, -0.2, 0.2, 0.5, 0.8$. We also vary the dimension $d$ and consider $d = 2, 3, 4, 5, 10.$  Two sample sizes $n = 100$ and $400$ are considered. For each scenario, we calculate the blockwise empirical likelihood ratio at block lengths ranging from $2$ to $16$ in order to examine the effects of block choices. In addition, we also use the progressive blocking method proposed by \cite*{Kim2013}, {\color{black} which is a blocking scheme that does not require the selection of block length}. For each scenario, $1000$ data sets are simulated and the likelihood ratio for each data set at the true mean is calculated. The coverage probability is then calculated as the proportion of times the likelihood ratio is less than the theoretical $\chi^2$ quantile at levels $\alpha = 0.1, 0.05, 0.01$. The likelihood ratios are calculated by the blockwise empirical likelihood without adjustment (BEL), adjusted blockwise empirical likelihood with $a = \log(n)/2$ (ABEL$_{\text{log}}$), $a = 0.5$ (ABEL\_0.5), $a = 0.8$ (ABEL$_{0.8}$), $a=1$(ABEL$_{1}$), and the high precision tuning parameter $a$ given by \eqref{eq: adj bart tuning} (ABEL$_{\text{hp}}$). The high precision tuning parameter \eqref{eq: adj bart tuning} is estimated by the plug-in estimator, which is then bias corrected by the blockwise bootstrap as described in Section \ref{sec: adj bart tune}. The full simulation results are shown in Table A1 in Appendix S2. Table \ref{table: adj coverage probability short} presented in this section shows a snapshot of Table A1 with AR(1) coefficients $\rho = -0.2, 0.2, 0.5, 0.8$. The block lengths $M$ shown in the tables are the ones that give the best coverage rates of each particular method, where $M=\text{pro}$ indicates that the progressive block method gives the best result. It can be seen that for negative $\rho$, BEL performed well and at least one of the considered ABEL methods matched or surpassed the BEL performance. As $\rho$ becomes positive, the BEL starts to show its vulnerability of under-coverage, and this becomes worse as dimension increases. In contrast, ABEL still provides adequate coverage. The phenomenon of BEL's under-coverage being not as severe for negative $\rho$ as for positive ones exemplifies the fact that the coverage probability is upper bounded by the probability of the convex hull containing the origin. For when $\rho$ is negative, the consecutive points are more likely to be on opposite sides in relation to the origin, and therefore the resulting convex hull is more likely to contain the origin and does not impose an upper bound on the coverage probability. Whereas, for positive $\rho$, the reverse effect happens. 

\begin{longtable}{lll||lrrr|lrrr}
   \caption{\small \em Comparison of Coverage Probabilities, $M = $  block length, $M=pro$ means progressive blocking method is used.} \\
\hline
\hline
&&&& n= & $100$ &&& n= & $400$ &\\
\hline
$\rho$ & $d$ & Methods & $M$ & 0.90 & 0.95 & 0.99 & $M$ & 0.90 & 0.95 & 0.99 \\ 
   \hline
-0.2 & 3 & BEL & 3 & 0.90 & 0.94 & 0.98 & 6 & 0.89 & 0.94 & 0.99 \\ 
  -0.2 & 3 & ABEL$_{\text{log}}$ & 3 & 0.94 & 0.97 & 0.99 & pro & 0.90 & 0.96 & 1.00 \\ 
  -0.2 & 3 & ABEL$_{0.8}$ & 3 & 0.91 & 0.95 & 0.98 & 7 & 0.90 & 0.94 & 0.99 \\ 
  -0.2 & 3 & ABEL$_{1}$ & 14 & 0.90 & 0.95 & 0.99 & 7 & 0.90 & 0.95 & 0.99 \\ 
  -0.2 & 3 & ABEL$_{\text{hp}}$ & 14 & 0.91 & 0.94 & 0.97 & pro & 0.90 & 0.95 & 0.99 \\ 
   \hline
   0.2 & 3 & BEL & 3 & 0.82 & 0.89 & 0.95 & 9 & 0.88 & 0.93 & 0.98 \\ 
  0.2 & 3 & ABEL$_{\text{log}}$ & 4 & 0.89 & 0.95 & 1.00 & 6 & 0.90 & 0.95 & 0.99 \\ 
  0.2 & 3 & ABEL$_{0.8}$ & 3 & 0.83 & 0.90 & 0.96 & 9 & 0.88 & 0.94 & 0.98 \\ 
  0.2 & 3 & ABEL$_{1}$ & 14 & 0.88 & 0.95 & 0.99 & 9 & 0.88 & 0.94 & 0.99 \\ 
  0.2 & 3 & ABEL$_{\text{hp}}$ & 12 & 0.93 & 0.96 & 0.98 & 8 & 0.90 & 0.95 & 1.00 \\ 
  \hline
  0.5 & 3 & BEL & 5 & 0.68 & 0.77 & 0.89 & 10 & 0.82 & 0.87 & 0.95 \\ 
  0.5 & 3 & ABEL$_{\text{log}}$ & 5 & 0.89 & 0.97 & 1.00 & 13 & 0.91 & 0.96 & 1.00 \\ 
  0.5 & 3 & ABEL$_{0.8}$ & 16 & 0.74 & 0.88 & 0.97 & 10 & 0.83 & 0.89 & 0.96 \\ 
  0.5 & 3 & ABEL$_{1}$ & 14 & 0.87 & 0.95 & 0.99 & 10 & 0.83 & 0.89 & 0.96 \\ 
  0.5 & 3 & ABEL$_{\text{hp}}$ & 14 & 0.92 & 0.95 & 0.97 & pro & 0.90 & 0.96 & 0.99 \\ 
     \hline
0.5 & 4 & BEL & 4 & 0.64 & 0.72 & 0.85 & 9 & 0.77 & 0.85 & 0.95 \\ 
  0.5 & 4 & ABEL$_{\text{log}}$ & 16 & 0.92 & 0.94 & 0.97 & 11 & 0.88 & 0.95 & 1.00 \\ 
  0.5 & 4 & ABEL$_{0.8}$ & 14 & 0.72 & 0.86 & 0.95 & 9 & 0.79 & 0.87 & 0.96 \\ 
  0.5 & 4 & ABEL$_{1}$ & 14 & 0.86 & 0.92 & 0.97 & 9 & 0.79 & 0.87 & 0.96 \\ 
  0.5 & 4 & ABEL$_{\text{hp}}$ & 13 & 0.91 & 0.94 & 0.96 & pro & 0.91 & 0.97 & 0.99 \\ 
  \hline
0.8 & 2 & BEL & 9 & 0.58 & 0.67 & 0.76 & 16 & 0.77 & 0.85 & 0.94 \\ 
  0.8 & 2 & ABEL$_{\text{log}}$ & 7 & 0.87 & 0.98 & 1.00 & 16 & 0.88 & 0.95 & 1.00 \\ 
  0.8 & 2 & ABEL$_{0.8}$ & 16 & 0.72 & 0.86 & 0.98 & 16 & 0.80 & 0.86 & 0.94 \\ 
  0.8 & 2 & ABEL$_{1}$ & 16 & 0.85 & 0.95 & 0.99 & 16 & 0.80 & 0.87 & 0.95 \\ 
  0.8 & 2 & ABEL$_{\text{hp}}$ & 4 & 0.91 & 0.94 & 0.97 & 13 & 0.90 & 0.96 & 1.00 \\ 
   \hline
   \hline
   \label{table: adj coverage probability short}
\end{longtable}

{\color{black}
\section{Application} \label{sec: applicatioin}
It is well known that the gross domestic product (GDP) is an important indicator that measures the health of the economy. GDP is one of the most important and widely reported economic data that are used by people from business owners to policymakers in their decision making \citep{Wolla2018HowGDP}. In this section, we demonstrate the ABEL method under a regression setting with GDP as the response variable and various economic indicators as covariates. The goal is to demonstrate how to apply ABEL in practice rather than to make new scientific discoveries. 
\begin{figure}
	\begin{center}
		\begin{tabular}{c}
			\includegraphics[width=16cm]{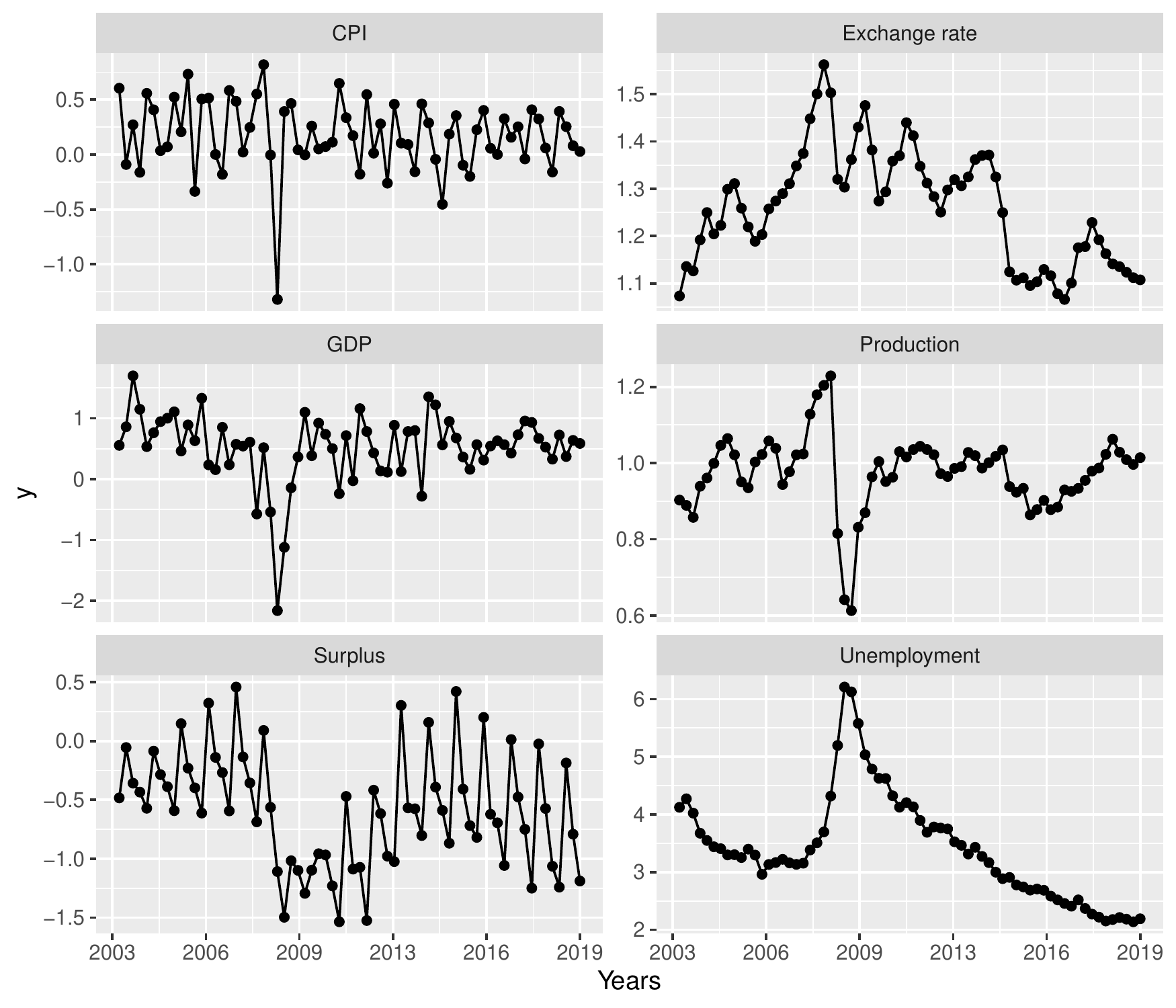}
		\end{tabular}
	\end{center}
	\caption{Plots of economic indicators: consumer price index (CPI), foreign exchange rate (Exchange rate), real gross domestic product (GDP), industrial production (Production), surplus or deficit (Surplus), unemployment rate (Unemployment).} 
	\label{sec app: gdp data}
	\vspace{-0.1in}
\end{figure}

The data were observed quarterly from 2003 to 2019 and were downloaded from St. Louis Federal Reserve website \url{https://fred.stlouisfed.org/}. As shown in Figure \ref{sec app: gdp data}, the data included consumer price index, units of growth rate previous period (CPI); U.S. / Euro foreign exchange rate, units of U.S. dollars to one Euro (Exchange rate); real gross domestic product, units of percent change (GDP); industrial production: manufacturing: durable goods: raw steel, units of Index $2017 = 1$ (Production); federal surplus or deficit, units of $100,100$ millions of dollars (Surplus); unemployment initial claims, units of $100,100$ number (Unemployment).
To illustrate inference with ABEL, we consider a linear model $Y_t =\beta_0 + \beta X_t + \epsilon_t$, where $ X_t \in \mathbb{R}^5$, $\beta_0 \in \mathbb{R}$, and $\beta^\top \in \mathbb{R}^5$ for predicting GDP with the $5$ predictors: CPI, Exchange rate, Production, Surplus, and Unemployment. For this linear model, we can use the estimating function $g(X_t, Y_t, \beta) = X_t (Y_t-\beta X_t)$ \citep{Owen2007EmpiricalModels, Kim2013}. Although the data plotted in Figure \ref{sec app: gdp data} appear to be nonstationary, the ABEL for estimating equations only require the estimating function $g(X_t, Y_t, \beta)$ to be stationary. To test the stationarity of $g(X_t, Y_t, \beta)$, we plug-in the ordinary least square estimate $\hat{\beta}$ for $\beta$ as suggested by \cite{Kim2013}. We then conducted Philips-Perrson test \citep{Phillips1988TestingRegression} by using the aTSA package \citep[cf. page 21 in][]{Qiu2015ATSA:Analysis} in R computing language \citep{RCoreTeam2020}, and the results indicated that $g(X_t, Y_t, \hat{\beta})$ were stationary. 

\begin{table}[h!]
\centering
\caption{Estimated coefficients by maximum blockwise empirical likelihood ($\hat{\beta}$),  $-2n/(QM)ABELR_Q(\beta_i)$ (ABEL), and $-2n/(QM)BELR(\beta_i)$ (BEL) under each null hypothesis $H_{i,0}, i = 1, \dotsc, 5$. BEL blows up to infinity because of the convex hull constraint. However, ABEL are still interpretable.} \label{sec app: regression table}
\begin{tabular}{ c|ccccc } 
\hline
\hline
Methods & CPI & Exchange rate & Production & Surplus & Unemployment \\
\hline
$\hat{\beta}$ & 0.38 & -1.09 & 0.93 & 0.10 & -0.08\\
ABEL & 0.60 & 12.22 & 15.86 & 0.77 & 0.39 \\ 
BEL & $270.70$ & $277.09$ & $277.18$& $275.41$ & $275.47$ \\ 
\hline
\hline
\end{tabular}
\end{table}

To see which economic indicators above are associated with GDP, we use Bonferroni correction to test the following hypothesis simultaneously, $H_{i,0}: \beta_i = 0$ v.s. $H_{i,1}: \beta_i \neq 0, i = 1, \dotsc, 5$. Under each null $H_{i,0}$, the ABEL ratio confidence region is constructed according to Theorem \ref{sec meth: adj asym chi under constraint}, where $r=1$ and $p=q=5$. For a familywise error rate of $0.05$, we have a confidence interval $\{\beta_i: -2n/(QM) ABELR_Q(\beta_i) \leq \chi^2_{1, 0.01} = 6.635 \}$ for each $\beta_i$. For comparison, we calculated BEL ratio confidence intervals according to Theorem 2 in \cite{Kitamura1997}, $\{\beta_i: -2n/(QM) BELR(\beta_i) \leq \chi^2_{1, 0.01} = 6.635 \}$. For ABEL, the high precision tuning parameter was estimated via the bias corrected plug-in estimator by the blockwise bootstrap as described in Section \ref{sec: adj bart tune}. For both ABEL and BEL, we used the progressive blocking scheme to achieve automatic block length selection \citep{Kim2013}. Table \ref{sec app: regression table} displays the calculated $-2n/(QM)ABELR_Q(\beta_i)$ (ABEL) and $-2n/(QM)BELR(\beta_i)$ (BEL) under each null hypothesis $H_{i,0}, i = 1, \dotsc, 5$. With ABEL, it can be seen that at a familywise error rate of $0.05$, it is failed to reject that the coefficients for CPI, surplus, and unemployment are different from 0; and coefficients for production and exchange rate are significantly different from 0. By contrast, due to the convex hull constraint exacerbated by the relatively small sample size, the empirical likelihood ratios by BEL were all very large indicating the nonconvergence of the algorithm solving for the likelihood ratios and thus rendered the results uninterpretable. 
}
\section{Conclusion}
We have shown that the newly proposed adjusted blockwise empirical likelihood is effective in improving the finite sample coverage probability for weakly dependent data. In particular, the corresponding ABEL ratio is shown to possess the asymptotic $\chi^2$ property similar to its non-adjusted counterpart. Moreover, we have shown that the adjustment tuning parameter can be used to achieve the asymptotic Bartlett corrected coverage error rate of $O(n^{-5/6})$ (i.e., the high precision tuning parameter). The formula for the high precision tuning parameter involves higher moments that need to be estimated in practice. We proposed to use the bias corrected plug-in estimator by blockwise bootstrap as a working solution, and it has been shown to be working well in our simulation and data application. However, we defer the systematic investigation of the high precision tuning parameter estimation to future studies. The simulation study showed that ABEL performed comparably to the non-adjusted BEL when the non-adjusted BEL performed well, and it outperformed the non-adjusted BEl by a large margin when the non-adjusted BEL suffered from the under-coverage issue. The effectiveness of ABEL was further demonstrated via an application to an economic data set. The moderate sample size and relative high parameter dimension rendered the test based on BEL uninterpretable; in contrast, the ABEL was still able to provide interpretable test results.

\section{Proofs} \label{sec: adj proofs}

\begin{proof}[Proof of Theorem \ref{thm: adj asymptotic chi}]
For ease of notation, we use $ T_i $ to denote $ T_i(\theta_0) $, where $ \theta_{0} $ is the true parameter. The first step in proving Theorem \ref{thm: adj asymptotic chi} is to show that the Lagrange multiplier $\lambda_a$ is $O_{p}(n^{-1/2}M)$, where we use the subscript $a$ to emphasis that this is the Lagrange multiplier in the adjusted blockwise empirical likelihood. First, we note that $\lambda_a$ solves the following equation 
\begin{equation} \label{eq: adj proof A1}
    \sum_{i=1}^{Q+1}\frac{T_{i}}{1+\lambda^{\top}T_{i}} = 0.
\end{equation}

Now, define $ \tilde{\lambda}_a := \lambda_a/\rho, \text{ where } \rho := \norm{\lambda_a} $. Multiply $ \tilde{\lambda}_a/Q $ on both sides of Equation \eqref{eq: adj proof A1}, and recall that $T_{Q+1} = -a \overline{T}$ (Equation \ref{eq: adj block extra point}), here $ \overline{T} $ means $ \overline{T}(\theta_{0}) $. Then we have
\begin{align}
	0 & = \frac{\tilde{\lambda}_a^{\top}}{Q} \sum_{i=1}^{Q+1} \frac{T_{i}}{1+\lambda_a^{\top} T_{i}} \nonumber \\
	& = \frac{\tilde{\lambda}_a^{\top}}{Q} \sum_{i=1}^{Q+1} T_{i} - \frac{\rho}{Q} \sum_{i=1}^{Q+1} \frac{(\tilde{\lambda}_a^{\top} T_{i})^{2}}{1+\rho \tilde{\lambda}_a^{\top} T_{i}} \nonumber \\
	& \leq \tilde{\lambda}_a^{\top} \overline{T} (1-\frac{a}{Q}) - \frac{\rho}{1+\rho T^{*}} \frac{1}{Q} \sum_{i=1}^{Q} (\tilde{\lambda}_a^{\top}T_{i})^{2} \nonumber \\
	& = \tilde{\lambda}_a^{\top} \overline{T} - \frac{\rho}{1+ \rho T^{*}} \tilde{\lambda}_a^{\top} \frac{1}{Q} \sum_{i=1}^{Q} T_{i} T_{i}^{\top} \tilde{\lambda}_a + O_{p}(n^{-1/2}Q^{-1}a). \label{eq: adj proof A2}
\end{align}
where $ T^{*} := \max_{1 \leq i \leq Q} \norm{T_{i}} $. By the law of large numbers, the central limit theorem for strong mixing processes (cf. \citeauthor{IBRAGIMOV1971}, \citeyear{IBRAGIMOV1971}), Lemma 3.2 in \cite{Kunsch1989} and assumption A.7 given in Section \ref{sec: adj bel}, \cite{Kitamura1997} has shown that $ T^{*} = o(n^{1/2}M^{-1}) $ a.s. Further, central limit theorem for strong mixing processes and assumption A.6 imply that $\sqrt{n}\overline{T} \rightarrow_d N(0, S)$. Then, we can deduce from \eqref{eq: adj proof A2} that 
\begin{equation*}
	0 \leq \tilde{\lambda}_a^{\top}\overline{T} - \frac{\rho}{M(1+\rho T^{*})} (1-\epsilon) \sigma_{1}^{2}(1 + o_P(1)) + O_{p}(n^{-1/2}Q^{-1}a),
\end{equation*}
where $0 < \epsilon <1$ and $\sigma_{1} > 0$ is the smallest eigenvalue of $S$. Then 
\begin{align*}
	&\frac{\rho}{M(1+\rho T_{*})} = O_{p} (n^{-1/2}Q^{-1} a) + O_{p}(n^{-1/2})\\
	\implies& \frac{\rho}{1+\rho T^{*}} = O_{p} (n^{-1/2}Q^{-1}M a) + O_{p}(n^{-1/2}M).
\end{align*}
	
By the assumption $ a = o_p(n/M) $ and $ QM \geq n $, we have
\begin{equation*}
    \frac{\rho}{1+\rho T^{*}} = O_{p} (n^{-1/2}M) \implies \rho = O_{p}(n^{-1/2} M).
\end{equation*}

Therefore, $ \lambda_a = O_{p}(n^{-1/2}M) $, which in particular means that $\lambda^{\top}_a T^* = o_p(1)$. The next step is to  express $ \lambda_a $ in terms of $ \overline{T} $. Notice that Equation \eqref{eq: adj proof A1} can be written as the sum of two parts
\begin{equation} \label{eq: fabel lambda rate 1}
	 0 = \frac{1}{Q} \sum_{i=1}^{Q+1} \frac{T_{i}}{1+\lambda_a^{\top}T_{i}} = \frac{1}{Q} \sum_{i=1}^{Q} \frac{T_{i}}{1+\lambda_a^{\top} T_{i}} + \frac{1}{Q} \frac{-a\overline{T}}{1-a\lambda_a^{\top}\overline{T}},
\end{equation}
where the first part on the right hand side can be written as 
\begin{align}
	& \frac{1}{Q} \sum_{i=1}^{Q} T_{i} \left[1- \lambda_a^{\top} T_{i} + \frac{(\lambda_a^{\top}T_i)^{2}}{1+\lambda_a^{\top}T_{i}}\right] \label{sec proof: op to Op 1} \\
	=& \frac{1}{Q} \sum_{i=1}^{Q} T_{i} - \frac{1}{Q} \sum_{i=1}^{Q} T_{i} \lambda_a^{\top} T_{i} + \frac{1}{Q} \sum_{i=1}^{Q} T_{i} \frac{(\lambda_a^{\top} T_{i})^{2}}{1+\lambda_a^{\top}T_{i}} \nonumber \\
	=& \overline{T} - \frac{1}{Q} \sum_{i=1}^{Q} T_{i} T_{i}^{\top} \lambda_a + o_{p}(n^{-1/2}). \nonumber
\end{align}
{ \color{black}
The last equality is because $ \max_{i} \|T_{i}\| = o_{p}(n^{1/2}M^{-1}) $ and $ |\lambda_{a}^{\top} T_{i}| \leq \|\lambda_{a}\| \max_{i} \|T_{i}\| = o_{p}(1) $ imply that $\lambda_{a}^{\top} T_{i} = o_{p}(1)$; and $  \frac{1}{Q} \sum_{i=1}^{Q} \|T_{i}\|^{3} = o_{p}(n^{1/2} M^{-1}) O_{p}(M^{-1}) $ since
\begin{align*}
	\frac{1}{Q} \sum_{i=1}^{Q} \|T_{i}\|^{3} \leq \max_{i} \|T_{i}\| \frac{1}{Q} \sum_{i=1}^{Q} \|T_{i}\|^{2} = o_{p}(n^{1/2} M^{-1}) O_{p}(M^{-1}). 
\end{align*}

Therefore, 
\begin{align}
\frac{1}{Q} \sum_{i=1}^{Q} T_{i} \frac{(\lambda_a^{\top} T_{i})^{2}}{1+\lambda_a^{\top}T_{i}} 
=& \frac{1}{Q} \sum_{i=1}^{Q} T_{i} (\lambda^{\top}_{a} T_{i})^{2} (1+\lambda_{a}^{\top} T_{i})^{-1}  \label{sec proof: op to Op} \\
\leq & \frac{1}{Q} \sum_{i=1}^{Q} \|T_{i}\|^{3} \|\lambda_{a}\|^{2} |1 + \lambda_{a}^{\top} T_{i}|^{-1} \nonumber \\
=&  o_{p}(n^{1/2} M^{-1}) O_{p}(M^{-1}) O_{p}(n^{-1}M^{2})  = o_{p}(n^{-1/2}). \nonumber
\end{align}
}

By using the assumption $a = o_{p}(n/M),$ the last summand in \eqref{eq: fabel lambda rate 1} is
\begin{align*}
	\frac{1}{Q} \frac{a\overline{T}}{1-a \lambda_a^{\top} \overline{T}} & = \frac{Q^{-1} o_{p}(nM^{-1})O_{p}(n^{-1/2})}{1-o_{p}(n/M) O_{p}(n^{-1}M)} \\
	& = \frac{o_{p}(n^{-1/2})}{1-o_{p}(1)}  \\
	& = o_{p}(n^{-1/2}) \text{ since } n \leq MQ.
\end{align*}
	
As a result, we have $ 0 =  \overline{T} - \frac{1}{Q} \sum_{i=1}^{Q} T_{i} T_{i}^{\top} \lambda_a + o_{p}(n^{-1/2})$. Then, we have the relationship
\begin{equation} \label{eq: adj proof A3}
	\lambda_a = M S^{-1}\overline{T} + o_{p}(n^{-1/2} M ).
\end{equation}

The final step is to Taylor expand the adjusted blockwise empirical likelihood ratio $ 2\frac{n}{MQ} \sum_{i=1}^{Q+1} \log(1+ \lambda_a^{\top} T_{i}) $. This ratio can be written as a sum of two parts
\begin{align*}
	 2\frac{n}{MQ} \sum_{i=1}^{Q+1} \log(1+\lambda_a^{\top}T_{i})  & = 2\frac{n}{MQ} \sum_{i=1}^{Q}\log(1+\lambda_a^{\top}T_{i}) + 2\frac{n}{MQ} \log(1+\lambda_a^{\top} T_{Q+1}),
\end{align*}
where the second part $ 2\frac{n}{MQ} \log(1+\lambda_a^{\top} T_{Q+1}) = o_{p}(1)$. This can be seen through a Taylor expansion, 
\begin{equation*}
    \log(1+\lambda_a^{\top} T_{Q+1}) = \lambda_a^{\top} T_{Q+1} - \frac{1}{2} (\lambda_a^{\top} T_{Q+1})^{2} + \eta ,
\end{equation*}
where for some finite $ B, P(|\eta| \leq B\norm{\lambda_a^{\top} T_{Q+1}}^{2}) \rightarrow 1 $. $ T_{Q+1} $ is defined as $ -a \overline{T} $, and from the first step in this proof, we know that $ \lambda_a = O_{p}(n^{-1/2}M) $. Therefore, $ \lambda_a ^{\top} T_{Q+1} = O_{p}(n^{-1/2}M) o(n/M)O_{p}(n^{-1/2}) = o_{p}(1) \text{ and } \eta = o_{p}(1)$.
	 
Now, a Taylor expansion of the first term gives 
\begin{align*}
	 2\frac{n}{MQ} \sum_{i=1}^{Q}\log(1+\lambda_a^{\top}T_{i}) & = 2\frac{n}{MQ} \sum_{i=1}^{Q}\left[\lambda_a^{\top} T_{i} - \frac{1}{2} (\lambda_a^{\top}T_{i})^{2} + \eta_{i}\right] \\
	 & \hspace*{-1.5cm}= 2nM^{-1} \lambda_a^{\top}\overline{T} - n M^{-1} \lambda_a^{\top} S M^{-1} \lambda_a + 2 \frac{n}{MQ} \sum_{i=1}^{Q} \eta_{i}\\
	 & \hspace*{-1.5cm}= 2n \overline{T}^{\top} S^{-1} \overline{T} -n\overline{T}^{\top} S^{-1} S S^{-1} \overline{T} + 2 \frac{n}{MQ} \sum_{i=1}^{Q} \eta_{i} + o_{p}(1)\\ 
	 &\hspace*{-1.5cm} = n \overline{T}^{\top} S^{-1} \overline{T} + o_{p}(1) \rightarrow_{d} \chi^{2}_{q},
\end{align*}
where $P\{|\eta_i| \leq \|\lambda_a^\top T_i\|^2 \} \rightarrow 1$.
This completes the proof.
\end{proof}

\begin{proof}[Proof of Theorem \ref{sec meth: adj asym chi under constraint}]
Without loss of generality, {\color{black} we assume that $\Theta_{r,0}$ is the parameter space under the following hypothesis,}
\begin{align} \label{sec meth: testing constraint 2}
    H_0: (\theta_{1}, \dotsc, \theta_{r})^\top = (\theta_{1, 0}, \dotsc, \theta_{r, 0})^\top .
\end{align}
Under this hypothesis, the unknown parameters need to be estimated are $\theta^c := (\theta_{r+1}, \dotsc, \theta_p)^\top$. In particular, the matrix $D^c := \mathbb{E} \partial g(x, \theta_0)/\partial (\theta^c)^\top$, where $\theta_0 := (\{\theta^H_0\}^\top, \{\theta_0^c\}^\top)^\top, \{\theta^H_0\}^\top = (\theta_{1,0}, \dotsc, \theta_{r,0})^\top,$ has rank $p-r$ because of assumption A.8. Then it is immediate from Theorem 1 in \cite{Kitamura1997}, that $\sqrt{n} (\hat{\theta^c} - \theta_0^c) \rightarrow_d N(0, \{(D^c)^\top S^{-1} D^c\}^{-1})$ , thus $\sqrt{n}\overline{T}(\hat{\theta}) \rightarrow_d N(0, S-D^c \{(D^c)^\top S^{-1} D^c\}^{-1} \{D^c\}^\top) $ \citep[cf. page 2098 in][]{Kitamura1997},
where $\hat{\theta}:=(\{\theta^H_0\}^\top, \{\hat{\theta^c}\}^\top)^\top$. Now by assumption A.7 and Lemma 3.2 in \cite{Kunsch1989}, following the argument of the proof of Theorem 1 in \cite{Owen1990} we have that $T^*(\hat{\theta}):=\max_{1\leq i \leq Q} \|T_i(\hat{\theta})\| = o(n^{1/2}M^{-1})$. Let $\tilde{\lambda}_a(\hat{\theta}) := \lambda_a(\hat{\theta}) / \rho$, where $\rho := \|\lambda_a(\hat{\theta})\|$. By similar calculations as in \eqref{eq: adj proof A1} and \eqref{eq: adj proof A2}, we have 
\begin{equation*}
    0 \leq \tilde{\lambda_a}(\hat{\theta})^\top \overline{T}(\hat{\theta}) - \frac{\rho}{1+\rho T^*(\hat{\theta})} \tilde{\lambda_a}(\hat{\theta})^\top \frac{1}{Q} \sum_{i=1}^Q T_i(\hat{\theta}) T_i(\hat{\theta})^\top \tilde{\lambda_a}(\hat{\theta}) + O_p(n^{-1/2}Q^{-1}a).
\end{equation*}
Since $M/Q \sum_i T_i(\hat{\theta}) T_i(\hat{\theta})^\top = S + o_p(1)$ \citep[cf. page 2097 in][]{Kitamura1997}, we have 
\begin{equation*}
    0 \leq \tilde{\lambda_a}(\hat{\theta})^\top \overline{T}(\hat{\theta}) - \frac{\rho}{M\{1+\rho T^*(\hat{\theta})\}}(1-\epsilon) \sigma_1^2\{1+o_p(1)\} + O_p(n^{-1/2}Q^{-1}a),
\end{equation*}
where $0 < \epsilon < 1$ and $\sigma_1 >0$ is the smallest eigenvalue of $S$. Then we have $\rho/\{1+\rho T^*(\hat{\theta})\} = O_p(n^{-1/2}M)$, thus $\rho = \|\lambda_a(\hat{\theta})\| = O_p(n^{-1/2}M)$. Next we notice that
\begin{equation*}
0=\frac{1}{Q} \sum_{i=1}^{Q+1} \frac{T_{i}(\hat{\theta})}{1+\lambda_{a}(\hat{\theta})^{\top} T_{i}(\hat{\theta})}=\frac{1}{Q} \sum_{i=1}^{Q} \frac{T_{i}(\hat{\theta})}{1+\lambda_{a}(\hat{\theta})^{\top} T_{i}(\hat{\theta})}+\frac{1}{Q} \frac{-a \overline{T}(\hat{\theta})}{1-a \lambda_{a}(\hat{\theta})^{\top} \overline{T}(\hat{\theta})},
\end{equation*}
where with similar calculations for Equations \eqref{sec proof: op to Op 1} and \eqref{sec proof: op to Op}, the first term is shown to be
\begin{equation*}
    \frac{1}{Q} \sum_{i=1}^{Q} \frac{T_{i}(\hat{\theta})}{1+\lambda_{a}(\hat{\theta})^{\top} T_{i}(\hat{\theta})} = \overline{T}(\hat{\theta}) - \frac{1}{Q} \sum_{i=1}^Q T_i(\hat{\theta}) T_i(\hat{\theta})^\top \lambda_a(\hat{\theta})+ o_p(n^{-1/2})
\end{equation*}
and the second term is
\begin{equation*}
\frac{1}{Q} \frac{a \overline{T}(\hat{\theta})}{1-a \lambda_{a}(\hat{\theta})^{\top} \overline{T}(\hat{\theta})} =o_{p} (n^{-1 / 2}).
\end{equation*}
As a result, we have
\begin{equation*}
    \lambda_a(\hat{\theta}) = MS^{-1}\overline{T}(\hat{\theta}) + o_p(n^{-1/2}M).
\end{equation*}
Finally, recall that $ \lambda_a(\hat{\theta}) ^{\top} T_{Q+1}(\hat{\theta}) = O_{p}(n^{-1/2}M) o(n/M)O_{p}(n^{-1/2}) = o_{p}(1)$. Writing
\begin{equation*}
    2\frac{n}{MQ} \sum_{i=1}^{Q+1} \log\{1+\lambda_a^{\top}(\hat{\theta})T_{i}(\hat{\theta})\} = 2\frac{n}{MQ} \sum_{i=1}^{Q} \log\{1+\lambda_a^{\top}(\hat{\theta})T_{i}(\hat{\theta})\} + 2\frac{n}{MQ} \log\{1+\lambda_a^{\top}(\hat{\theta})T_{Q+1}(\hat{\theta})\},
\end{equation*}
a Taylor expansion applied to the last term gives 
\begin{align*}
    2\frac{n}{MQ} \log\{1+\lambda_a(\hat{\theta})^{\top}T_{Q+1}(\hat{\theta})\} &= 2\frac{n}{MQ}[ \lambda_a(\hat{\theta})^\top T_{Q+1}(\hat{\theta}) - \frac{1}{2}\{\lambda_a(\hat{\theta})^\top T_{Q+1}(\hat{\theta})\}^2 + \eta] = o_p(1),
\end{align*}
with $\eta$ satisfying that $P\{|\eta| \leq \|\lambda_a(\hat{\theta})^\top T_{Q+1}(\hat{\theta})\|^2 \} \rightarrow 1$.
A Taylor expansion applied to first term results in 
\begin{align*}
    2\frac{n}{MQ} \sum_{i=1}^{Q} \log\{1+\lambda_a^{\top}(\hat{\theta})T_{i}(\hat{\theta})\} & = 2\frac{n}{MQ} \sum_{i=1}^{Q}\left[\lambda_a(\hat{\theta})^{\top} T_{i}(\hat{\theta}) - \frac{1}{2} \{\lambda_a(\hat{\theta})^{\top}T_{i}(\hat{\theta})\}^{2} + \eta_{i}\right] \\
	 & \hspace*{-1.5cm}= 2nM^{-1} \lambda_a(\hat{\theta})^{\top}\overline{T}(\hat{\theta}) - n M^{-1} \lambda_a(\hat{\theta})^{\top} S M^{-1} \lambda_a(\hat{\theta}) + 2 \frac{n}{MQ} \sum_{i=1}^{Q} \eta_{i}\\
	 & \hspace*{-1.5cm}= 2n \overline{T}(\hat{\theta})^{\top} S^{-1} \overline{T}(\hat{\theta}) -n\overline{T}(\hat{\theta})^{\top} S^{-1} S S^{-1} \overline{T}(\hat{\theta}) + 2 \frac{n}{MQ} \sum_{i=1}^{Q} \eta_{i} + o_{p}(1)\\ 
	 &\hspace*{-1.5cm} = n \overline{T}(\hat{\theta})^{\top} S^{-1} \overline{T}(\hat{\theta}) + o_{p}(1),
\end{align*}
where $P\{|\eta_i| \leq \|\lambda_a(\hat{\theta})^\top T_i(\hat{\theta)}\|^2 \} \rightarrow 1$.
As a consequence, we have 
\begin{align*}
    2\frac{n}{MQ} \sum_{i=1}^{Q+1} \log[1+\lambda_a^{\top}(\hat{\theta})T_{i}(\hat{\theta})] &= n \overline{T}(\hat{\theta})^{\top} S^{-1} \overline{T}(\hat{\theta}) + o_p(1) \rightarrow_d \chi_{q-p+r}^2,
\end{align*}
where the last step is because $S^{-1/2}\sqrt{n} \overline{T}(\hat{\theta}) \rightarrow_d N(0, S^{-1/2}[S-D^c \{(D^c)^\top S^{-1} D^c\}^{-1} \{D^c\}^\top] S^{-1/2})$, and it is a straightforward calculation to see that $S^{-1/2}[S-D^c \{(D^c)^\top S^{-1} D^c\}^{-1} \{D^c\}^\top] S^{-1/2}$ is idempotent with rank $q-p+r$. This completes the proof.
\end{proof}

\begin{proof}[Proof of Theorem \ref{thm: high order adj coverage}]
The foundation of the Bartlett corrected coverage probability of the empirical likelihood ratio confidence region was laid out by \cite*{DiCiccio1988} for i.i.d data. The major steps include (1) finding the signed-root decomposition of the log empirical likelihood ratio, (2) measuring the size of the third and fourth joint cumulants of the signed-root, and (3) applying the Edgeworth expansion of the density of the signed-root. For weakly dependent data, \cite{Kitamura1997} obtained the Bartlett corrected coverage error for blockwise empirical likelihood ratio confidence region under the smooth function model by following the above major steps with modifications to account for data dependency. For adjusted empirical likelihood, \cite{Liu2010} exploited the fact that the tuning parameter features in the signed-root of the adjusted empirical likelihood ratio and used it as a leverage to eliminate the large error terms to achieve a Bartlett corrected error rate. Here we will synthesize the above techniques to show that ABEL ratio confidence region under the general estimating framework can achieve higher order coverage accuracy for weakly dependent data. \textcolor{black}{As in \cite{Kitamura1997}, we assume that all required higher moments exist in this proof.}

The first step is to derive the signed-root decomposition of $ ABELR_{Q}(\theta_{0}) $. For this, we need to first establish the relationship between the Lagrange multiplier $ \lambda_{a} $ in the adjusted blockwise empirical likelihood (Equation \ref{eq: adj blockwise el}) and the Lagrange multiplier $ \lambda $ in the non-adjusted blockwise empirical likelihood \citep[cf. Equation 3.5 in][]{Kitamura1997}. Let 
\begin{equation*}
	 f(\zeta) := \frac{1}{Q}\sum_{i = 1}^{Q} \frac{T_i}{1+\zeta^{\top}T_i}.
\end{equation*}
	
Then, by definition, we have $ f(\lambda) = 0 $, and the adjusted Lagrange multiplier $ \lambda_{a} $ satisfies 
\begin{align} \label{sec proof: flambda and T bar}
	f(\lambda_{a}) = \frac{1}{Q} a \overline{T} + O_{p}(n^{-3/2}Q^{-1}M)
\end{align}
because
\begin{align*}
0 & = \frac{1}{Q} \sum_{i = 1}^{Q+1} \frac{T_i}{1 + \lambda_{a} T_i} = \frac{1}{Q} \sum_{i = 1}^{Q} \frac{T_i}{1+\lambda_{a}^{\top}T_i} + \frac{1}{Q} \frac{-a \overline{T}}{1- \lambda_{a}^{\top} a \overline{T}} \\
& = f(\lambda_{a}) - \frac{1}{Q} a \overline{T} + O_{p}(n^{-3/2}Q^{-1}M),
\end{align*}
where, with the facts that $ \overline{T} = O_{p}(n^{-1/2}) $ \citep{Kitamura1997} and $ \lambda_{a} = O_{p}(n^{-1/2}M) $ (Proof of Theorem \ref{thm: adj asymptotic chi}), the last term is
\begin{align*}
\frac{1}{Q} \frac{-a \overline{T}}{1- \lambda_{a}^{\top} a \overline{T}}  &= -\frac{a\overline{T}}{Q} [1 + \zeta ], \text{ where } |\zeta| \leq |\lambda_{a}^{\top} \overline{T}| \\
&= -\frac{1}{Q} a\overline{T} + O_{p}(n^{-3/2}Q^{-1}M).
\end{align*}
	
Next, we show that
\begin{align} \label{sec proof: t bar and lamabda relation}
\overline{T} = M^{-1}\lambda + O_{p}(Q^{-1/2}M^{-1}).
\end{align}

To see this, notice that with details given in the next Section S1, we have
\begin{align} \label{sec proof: clt 1}
\frac{1}{Q} \sum_{i=1}^{Q} T_{i} \frac{(\lambda^{\top} T_{i})^{2}}{1+\lambda^{\top}T_{i}} 
= O_{p}(n^{-1}M^{1/2}),
\end{align}
and
\begin{align} \label{sec proof: sample variance clt}
\frac{1}{Q} \sum_{i=1}^{Q} T_{i}T_{i}^{\top} = M^{-1}I + O_p(Q^{-1/2}M^{-1}).
\end{align}

Thus, with \eqref{sec proof: clt 1} and \eqref{sec proof: sample variance clt}, we have
\begin{align*} 
0  &= \frac{1}{Q} \sum_{i=1}^{Q} \frac{T_{i}}{1+\lambda^{\top} T_{i}}\\ 
&= \frac{1}{Q} \sum_{i=1}^{Q} T_{i} \left[1- \lambda^{\top} T_{i} + \frac{(\lambda^{\top}T_i)^{2}}{1+\lambda^{\top}T_{i}}\right] \nonumber \\
&= \frac{1}{Q} \sum_{i=1}^{Q} T_{i} - \frac{1}{Q} \sum_{i=1}^{Q} T_{i} \lambda^{\top} T_{i} + \frac{1}{Q} \sum_{i=1}^{Q} T_{i} \frac{(\lambda^{\top} T_{i})^{2}}{1+\lambda^{\top}T_{i}} \nonumber \\
&= \overline{T} - \frac{1}{Q} \sum_{i=1}^{Q} T_{i} T_{i}^{\top} \lambda+ O_{p}(n^{-1}M^{1/2}) \nonumber \\
&= \overline{T} - M^{-1} \lambda+ O_{p}(Q^{-1/2}M^{-1}). \nonumber
\end{align*}

Therefore, we have \eqref{sec proof: t bar and lamabda relation}, and plug it into (\ref{sec proof: flambda and T bar}) we have
\begin{align} \label{sec proof: flambda and T bar2}
f(\lambda_{a}) = \frac{1}{QM} a \lambda + O_{p}(Q^{-3/2}M^{-1}).
\end{align}

Next, by Taylor expansion, we have
\begin{equation*}
	f(\lambda_{a}) = f(\lambda) +\frac{\partial f(\lambda)}{\partial \lambda} (\lambda_{a} - \lambda) + O_p((\lambda_{a}-\lambda)^{2}),
\end{equation*}
which, since $ f(\lambda) = 0 $, gives
\begin{align}
	\lambda_{a}- \lambda & = \left(\frac{\partial f(\lambda)}{\partial \lambda}\right)^{-1} f(\lambda_{a}) + O_p((\lambda_{a}-\lambda)^{2}). \label{eq: adj proof B1}
\end{align}
	
Then we need to find $ \partial f(\lambda)/\partial \lambda $.  By assuming that higher moments exist, then by similar arguments for Equation \eqref{sec proof: clt 1}, we have
\begin{align*}
	\frac{\partial f(\lambda)}{\partial \lambda} & = -\frac{1}{Q} \sum_{i = 1}^{Q} \frac{T_iT_i^{\top}}{[1+\lambda^{\top}T_i]^{2}}\\
	& = -\frac{1}{Q} \sum_{i = 1}^{Q} T_i T_i^{\top} \left[ 1 - \lambda^{\top} T_i + \frac{(\lambda^{\top} T_i)^{2}}{1 + \lambda^{\top} T_i} \right]^{2}\\
	& = -\frac{1}{Q} \sum_{i = 1}^{Q} T_i T_i^{\top} \left[ 1 + (\lambda^{\top} T_i)^{2} + \frac{(\lambda^{\top} T_i)^{4}}{(1 + \lambda^{\top} T_i)^{2}} \right.\\
	& \left.- 2\lambda^{\top}T_i + 2\frac{(\lambda^{\top} T_i)^{2}}{1 + \lambda^{\top} T_i} - 2\frac{(\lambda^{\top} T_i)^{3}}{1 + \lambda^{\top} T_i} \right]\\
	& = - \frac{1}{Q} \sum_{i = 1}^{Q} T_i T_i^{\top} + O_{p}(n^{-1/2}M^{-1/2}).
\end{align*}

Then by \eqref{sec proof: sample variance clt}, 
\begin{align*}
	\frac{\partial f(\lambda)}{\partial \lambda} = -M^{-1} I + O_p(Q^{-1/2}M^{-1}).
\end{align*}
		
Now plug $ \partial f(\lambda)/\partial \lambda $ and \eqref{sec proof: flambda and T bar2} into \eqref{eq: adj proof B1} to get 
\begin{equation} \label{eq: adj proof B2.1}
	\lambda_{a} - \lambda = -\frac{a}{Q}  \lambda + O_{p}(Q^{-3/2}). 
\end{equation}

Then, Equation \eqref{eq: adj proof B2.1} gives 
\begin{equation} \label{eq: adj proof B3}
	\lambda_{a} = \left(1 - \frac{a}{Q} \right) \lambda + O_{p}(Q^{-3/2}).
\end{equation}	

Now, substitute $ \lambda_{a} $ into the $ ABELR_{n}(\theta_{0}) $ (\ref{eq: adj blockwise elr}), and with details for Equation \eqref{sec proof: plugin lambda a} given in the next Section S1,  we have
\begin{align}
	\frac{n}{MQ} ABELR_{n}(\theta_{0}) & = 2 \frac{n}{MQ} \sum_{i = 1}^{Q + 1} \log(1 + \lambda_{a}^{\top} T_i) \nonumber \\
	& = 2 \frac{n}{MQ} \sum_{i = 1}^{Q+1} \log \left[1 + \left(1 - \frac{a}{Q}\right) \lambda^{\top} T_i\right] + O_{p}(n^{-5/6}) \label{sec proof: plugin lambda a} \\
	& = 2 \frac{n}{MQ} \sum_{i = 1}^{Q} \log \left[1 + \left(1 - \frac{a}{Q}\right) \lambda^{\top} T_i\right] \label{eq: adj proof B4.1}  \\
	& + 2 \frac{n}{MQ} \log \left[1 + \left(1 - \frac{a}{Q}\right) \lambda^{\top} T_{n+1}\right] \label{eq: adj proof B4.2} \\
	& + O_{p}(n^{-5/6}) \nonumber.
\end{align}

Again with details given in the next Section S1, \eqref{eq: adj proof B4.1} can be written as
\begin{align} \label{sec proof: separate a}
2 \frac{n}{MQ} \sum_{i = 1}^{Q} \log \left[1 + \left(1 - \frac{a}{Q}\right) \lambda^{\top} T_i\right] 
&=  2 \frac{n}{MQ} \sum_{i = 1}^{Q} \log (1 + \lambda^{\top} T_i) \\
&+ O_{p}(n^{-5/6}). \nonumber
\end{align}

Now for \eqref{eq: adj proof B4.2}, we have
\begin{align*}
	& 2 \frac{n}{MQ} \log \left[1 + \left(1 - \frac{a}{Q}\right) \lambda^{\top}T_{n+1}\right] \\
	&=  2 \frac{n}{MQ} \left[\lambda^{\top}T_{n+1} - \frac{a}{Q} \lambda^{\top} T_{n+1} - \frac{1}{2} \zeta^2 \right], \text{ where } |\zeta| \leq \left| \left(1-\frac{a}{Q} \right) \lambda^{\top} T_{n+1} \right|\\
	&=  - 2 \frac{an}{MQ} \lambda^{\top} \overline{T} + O_{p}(Q^{-2})\\
	&=  -2 \frac{an}{M^{2}Q} \lambda^{\top} \lambda + O_{p}(n^{-5/6}),
\end{align*}
where the last equality is due to \eqref{sec proof: t bar and lamabda relation}. As a result, we have
\begin{equation*}
	\frac{n}{MQ} ABELR_{Q}(\theta_{0}) = 2 \frac{n}{MQ} \sum_{i = 1}^{Q} \log(1 + \lambda^{\top} T_{i})  -2 \frac{an}{M^{2}Q} \lambda^{\top} \lambda + O_{p}(n^{-5/6}).
\end{equation*}

{\color{black}
Now we are ready to derive the signed-root decomposition of $ ABELR_{Q}(\theta_{0}) $. To aid our derivation, we first introduce some further notations. Let
\begin{equation} \label{sec proof: A}
A^{j_{1} \cdots j_{v}} = \frac{M^{v-1}}{Q} \sum_{i = 1}^{Q} T_{i}^{j_{1}} \cdots T_{i}^{j_{v}} - \alpha^{j_{1} \cdots j_{v}},
\end{equation}
where $ \alpha^{j_{1} \cdots j_{v}} $ is defined by Equation \eqref{eq: adj moment notation} in Section \ref{sec: adj bart tune} of the main paper. By following step 2 on page 7 in \cite*{DiCiccio1988} with their $ t $ and $ X_{i} $ replaced by our $ \lambda $ and $ T_{i} $ respectively, and notice that in our weakly dependent setting, $ \lambda = O_{p}(n^{-1/2}M) $, we have the following decomposition
\begin{align} \label{eq: adj sign-root}
2 \frac{n}{MQ} \sum_{i = 1}^{Q} \log(1 + \lambda^{\top} T_{i}) &= 2n(R_{1}+R_{2}+R_{3})^{\top}(R_{1}+ R_{2} + R_{3}) + O_{p}(n^{-4/3}),
\end{align}
where, for $r, s, t = 1, \dotsc, q$
\begin{align} \label{sec proof: R}
& R_{1}^{r} = A^{r}\\
& R_{2}^{r} = \frac{1}{3} \alpha^{rst} A^{s} A^{\top} - \frac{1}{2} A^{rs} A^{s}\\
& R_{3}^{r} = \frac{3}{8} A^{rs} A^{st} A^{\top} - \frac{5}{12} \alpha^{rst} A^{tu} A^{s} A^{u} - \frac{5}{12} \alpha^{stu} A^{rs} A^{\top} A^{u} \\
& + \frac{4}{9} \alpha^{rst} \alpha^{tuv} A^{s} A^{u} A^{v} + \frac{1}{3} A^{rst} A^{s} A^{\top} - \frac{1}{4} \alpha^{rstu} A^{s} A^{\top} A^{u}.
\end{align}
	
Here the summation over repeated index is used, for example $ A^{rs}A^{s} = \sum_{s=1}^{q} A^{rs} A^{s} $. The order $ O_{p}(n^{-4/3}) $ in Equation \eqref{eq: adj sign-root} can be found by using similar arguments for Equation (S.1.1) given in the next Section S1. We use the facts that $ \sum_{i=1}^{Q} (\lambda^{\top} T_{i})^{5}= O_{p}(\|\lambda\|^{5} Sd(\sum_{i=1}^{Q} T_{i}^{j_{1}} \dotsc T_{i}^{j_{5}}))$, where $ Sd $ means standard deviation. We know that $ \lambda =O_{p}(n^{-1/2}M) $, and as in the derivation for Equation (S.1.1), we can show that $ Sd(\sum_{i=1}^{Q} T_{i}^{j_{1}} \dotsc T_{i}^{j_{5}}) = O_{p}(Q^{1/2}M^{-5/2}) $. This gives
\begin{align*}
\sum_{i=1}^{Q} \lambda^{\top} T_{i}^{5} &= O_{p}(n^{-5/2}M^{5} Q^{1/2} M^{-5/2})  \\ 
& = O_{p}(n^{-5/2} M^{5/2} Q^{5/2} Q^{-2})  \\ 
& = O_{p}(Q^{-2})  \tag{using $ MQ = O(n) $} \\
& = O_{p}(n^{-4/3}).
\end{align*}
	
Since we add an extra blockwise estimating Equation \eqref{eq: adj block extra point} in the adjusted blockwise empirical likelihood $ ABELR_Q(\theta_0) $, the signed-root decomposition will be slightly affected by the adjustment. And this is exactly where we can leverage the tuning parameter to achieve Bartlett corrected coverage error rate. With details given in the online supplement Section S1, we have the signed-root decomposition of $ ABELR_{Q}(\theta_{0}) $ as
}
\begin{align} \label{sec proof: signed root for abelr}
	& \frac{n}{MQ} ABELR_{Q}(\theta_{0}) \\
	=& 2 n(R_{1} + R_{2} + R_{3} -\frac{a}{Q} R_{1})^{\top}(R_{1} + R_{2} + R_{3} -\frac{a}{Q} R_{1}) + O_{p}(n^{-5/6}), \nonumber
\end{align}
where the $2 \frac{an}{M^{2}Q} \lambda^{\top} \lambda$ term is factored in the $\frac{a}{Q} R_1$ term. \\

The next step is to derive the joint cumulants of $ \sqrt{n}R := \sqrt{n}(R_{1} + R_{2} + R_{3} -\frac{a}{Q} R_{1})$. 

{ \color{black}	
Let $ \kappa^{r}, \kappa^{ri}, \kappa^{uvw}, \text{ and } \kappa^{ruvw}$ denote the first 4 joint cumulants of $ \sqrt{n}R $. With the assumed mixing rate $ \alpha_{X}(m) \leq c e^{-dm} $ as given in Theorem \ref{thm: high order adj coverage}, it can be seen that $ T_{i} $ is an asymptotically 1-dependent process with asymptotic approximation error $ O(e^{-d(k-1)/n^{1/3}}) $ because 
\begin{align} \label{sec proof: 1-dependent}
	\alpha_{T}(k) & \leq \alpha_{X}((k-1)M) \leq c e^{-d(k-1)M} \\
	& = c e^{-d(k-1)n^{1/3}}. \nonumber
\end{align}

This observation plays an important role in deriving the cumulants. In particular, it implies the formulas given in Step 6 in \cite*{DiCiccio1988} with their $ H^{j} $ and $ Z^{j} $ replaced by $ X $ and $ n^{-1} \sum X_{i} $ respectively, and with the correlations taken into account by $ \tilde{\alpha} $ defined by Equation \eqref{sec tuning: alpha tilde} in Section \ref{sec: adj bart tune} of the main paper. Let $ \kappa^{r}, \kappa^{ri}, \kappa^{uvw} \text{ and } \kappa^{ruvw}$ denote the first 4 joint cumulants of $ \sqrt{n}R $. 

We first show that $ \kappa^{uvw} = O_{p}(n^{-5/6})$. With the general formula of the third cumulant given in \cite{McCullagh1987TensorStatistics}  \citep*[see also step 7 in][]{DiCiccio1988}, we have:
\begin{align} \label{sec proof: third cumulant def}
	n^{-3/2} \kappa^{uvw} = \mathbb{E} (R^{u} R^{v} R^{w}) - \mathbb{E} (R^{u}) \mathbb{E} (R^{v} R^{w})[3] + 2 \mathbb{E} (R^{u}) \mathbb{E} (R^{v}) \mathbb{E} (R^{w}),
\end{align}
where $ \mathbb{E} (R^{u}) \mathbb{E} (R^{v} R^{w})[3] $ denotes a sum of $ 3 $ terms that permutes over the superscripts. For example, $ \mathbb{E} (R^{u}) \mathbb{E} (R^{v} R^{w}) + \mathbb{E} (R^{v}) \mathbb{E} (R^{u} R^{w}) + \mathbb{E} (R^{w}) \mathbb{E} (R^{u} R^{v}) $. Now plug $ R = R_{1} + R_{2} + R_{3} - a/Q R_{1} $ into Equation \eqref{sec proof: third cumulant def} and with $ T_{i} $ being 1-dependent as shown above, then after lengthy but routine calculations using formulas in Step 6 in \cite*{DiCiccio1988} that are appropriately modified as mentioned above, we have 
\begin{align} \label{sec proof: 3rd cumulant all parts}
	n^{-\frac{3}{2}} \kappa^{uvw} & = \left( 1 - \frac{a}{Q} \right)^{2} \{ \mathbb{E}(R_{1}^{u}R_{1}^{v}R_{1}^{w}) +\mathbb{E}(R_{2}^{u}R_{1}^{v}R_{1}^{w})[3] \\
	& \hspace{1.7in} - \mathbb{E}(R_{2}^{u}) \mathbb{E}(R_{1}^{v}R_{1}^{w}) [3]
	+ O(Q^{-3}M^{-2}) \}. \nonumber \\
	& = \mathbb{E}(R_{1}^{u}R_{1}^{v}R_{1}^{w}) +\mathbb{E}(R_{2}^{u}R_{1}^{v}R_{1}^{w})[3] - \mathbb{E}(R_{2}^{u}) \mathbb{E}(R_{1}^{v}R_{1}^{w}) [3] + O(Q^{-3}M^{-2}). \nonumber
\end{align}

With detailed calculations given in the Section S1, we have 
\begin{align} \label{sec proof: 3rd cumulant part 1}
	\mathbb{E}(R_{1}^{u}R_{1}^{v}R_{1}^{w}) &= M^{-2} Q^{-2} \tilde{\alpha}^{uv, w}[3] - 2M^{-2}Q^{-2} \alpha^{uvw} \\
	& + O(Q^{-2} M^{-2} e^{-d(1-r)n^{1/3}/2}),\ r <1, \nonumber
\end{align}
\begin{align} \label{sec proof: 3rd cumulant part 2}
	\mathbb{E}(R_{2}^{u}R_{1}^{v}R_{1}^{w}) &= \frac{1}{3} Q^{-2}M^{-2}\alpha^{uss} \delta^{vw} + \frac{2}{3} Q^{-2} M^{-2} \alpha^{uvw} \\
	& - \frac{1}{2} Q^{-2}M^{-2} \tilde{\alpha}^{us, s} \delta^{vw}  \nonumber \\
	& - \frac{1}{2} Q^{-2}M^{-2} \tilde{\alpha}^{uv, w} - \frac{1}{2} Q^{-2}M^{-2} \tilde{\alpha}^{uw, v} \nonumber \\
	& + O(Q^{-3}M^{-1}),  \nonumber
\end{align}
and 
\begin{align} \label{sec proof: 3rd cumulant part 3}
	\mathbb{E}(R_{2}^{u}) \mathbb{E}(R_{1}^{v}R_{1}^{w}) &= \frac{1}{3} Q^{-2}M^{-2}\alpha^{uss} \delta^{vw} - \frac{1}{2} Q^{-2}M^{-2} \tilde{\alpha}^{us, s} \delta^{vw}
\end{align}

Now plug Equations \eqref{sec proof: 3rd cumulant part 1}, \eqref{sec proof: 3rd cumulant part 2} and \eqref{sec proof: 3rd cumulant part 3} into Equation \eqref{sec proof: 3rd cumulant all parts}, we have 
\begin{align}
	\kappa^{uvw} = O(Q^{-3} M^{-1}) = O(n^{-5/6}).
\end{align}

With similar calculations as above, it can be shown the fourth cumulant is $ \kappa^{ruvw} = O(n^{-5/6}) $. 
}

The first culumant is the first moment, so $ \kappa^{r} := \sqrt{n} E(R_{1}^{r} + R_{2}^{r} + R_{3}^{r} -\frac{a}{Q} R_{1}^{r}) $, where $ 	\mathbb{E} (R_{1}^{r}) = 0 $, $ \mathbb{E} (R_{2}^{r}) = \left(1/3 \alpha^{rss} - 1/2 \tilde{\alpha}^{rs, s} \right) n^{-1}$ (shown in Section S1 Equation S1.4), and $ 	\mathbb{E} (R_{3}^{r}) = O(M^{-1} Q^{-2}) $. The first cumulant is then
\begin{equation*}
    \kappa^{r} = \left(\frac{1}{3} \alpha^{rss}  - \frac{1}{2} \tilde{\alpha}^{rs, s} \right) n^{-1/2} + O(n^{-1/2}Q^{-1}).
\end{equation*}	
	
The second cumulant $ \kappa^{ri} $ is more complex, but fortunately the adjusted signed-root $ R $ is different from the non-adjusted signed-root $ \mathcal{R} := R_{1} +R_{2} + R_{3} $ by just a factor of $ \frac{a}{Q} R_{1} $. This will allow us to modify the Bartlett factor formula in \cite{Kitamura1997} for $ \mathcal{R} $ because $ R = \mathcal{R} - \frac{a}{Q}R_{1}$. With this notation, the second cumulant can be written in a relatively simple form as
\begin{align*}
	\kappa^{ri} & := n Cov(R^{r}, R^{i}) \\
	& = n Cov(\mathcal{R}^{r}, \mathcal{R}^{i}) - 2n Cov\left(\mathcal{R}^{r}, \frac{a}{Q} R_{1}^{i}\right) + n Cov \left( \frac{a^{2}}{Q^{2} } R_{1}^{r}, R_{1}^{i} \right).
\end{align*}
	
For the first term, we have
\begin{align*}
	nCov(\mathcal{R}^{r}, \mathcal{R}^{i}) & = n E(\mathcal{R}^{r} \mathcal{R}^{i}) - nE(\mathcal{R}^{r}) E(\mathcal{R}^{i}) \\
	& = nE(R_{1}^{r} R_{1}^{i}) + n^{-1} a_{ri} + O(n^{-5/6}) \\
	& =  \delta_{ri} + n^{-1} a_{ri} - n^{-1} D + O(n^{-5/6}) \\
	& = \delta_{ri} + n^{-1} ( a_{ri} - D )+ O(n^{-5/6}), 
\end{align*}
where $ a_{ri} $ is given in Section \ref{sec: adj bart tune} Equation \eqref{eq: adj ari} and $ D := (1/3 \alpha^{iss} - 1/2\tilde{\alpha}^{is,s}) (1/3 \alpha^{jss} - 1/2\tilde{\alpha}^{js,s})$. For the second and the third term, it can be shown by the calculations given in Section S1 that
\begin{equation} \label{eq: proof cummu 2}
    nCov(\mathcal{R}^{r}, \frac{a}{Q} R_{1}^{i}) = \frac{a}{Q} \delta_{ri} + O(n^{-1}Q^{-1})
\end{equation}
and 
\begin{equation} \label{eq: proof cummu 3}
    nCov(\frac{a^2}{Q^2} R_{1}^{r},  R_{1}^{i}) = O(n^{-4/3}).
\end{equation}

As a result, 
\begin{align*}
	\kappa^{ri} = \delta_{ri} + \left[a_{ri} - D - 2 a \delta_{ri} \frac{n}{Q}\right] n^{-1} + O(n^{-5/6}).
\end{align*}

Notice that $ \sqrt{n}R $ is a smooth vector-valued function of $ A^{j_{1} \cdots j_{v}} $, which are centered sample moments of blockwise data $ T_{i} $. According to \cite{Davison1993}, the validity of Edgeworth expansion for blocks of weakly dependent data can be established under conditions mentioned in the beginning of Section \ref{sec: adj bart tune} by following similar arguments given in \cite{Bhattacharya1978}. \cite{Kitamura1997} used similar arguments for blockwise empirical likelihood. See also \cite{Lahiri1991SecondBootstrap} and \cite{Lahiri1996OnModels} for the validity of Edgeworth expansion with blockwise weekly dependent data. The Edgeworth expansion can be obtained through its formal Edgeworth expansion using the cumulants. The number of cumulants needed depends on the order in the expansion. For the order $ O(n^{-5/6}) $ in our expansion, we need the first four cumulants, which are calculated above. 

With detailed calculations given in Section S1, we have the following formal Edgeworth expansion for the density of 
$ \sqrt{n} R $ as
\begin{equation} \label{sec proof: formal edgeworth expansion}
    f_{R} (x) = \phi(x) + n^{-1/2} W_{1}(x) + n^{-1} W_{2}(x) + O(n^{-5/6}),
\end{equation}
where 	
\begin{align*}
	W_{1}(x) & = \left(\frac{1}{3} \alpha^{iss}  - \frac{1}{2} \tilde{\alpha}^{is, s}\right) x^{i} \phi(x)
\end{align*}
and
\begin{equation*}
    W_{2}(x) = \frac{1}{2} \left(a_{ij} -  2 a \delta_{ij} \frac{n}{Q} \right) (x^{i} x^{j} - \delta_{ij}) \phi(x)
\end{equation*}
with $ a_{ij} $ given by Equation \eqref{eq: adj ari} in Section \ref{sec: adj bart tune} of the main paper and $ \delta_{ij} $ being the Kronecker delta. Notice that $W_2(x) = O(M)$, thus $n^{-1} W_2(x) = O(n^{-2/3})$. With the above expansion, we then have
\begin{equation*}
    P(nR^{\top}R \leq x) = \int_{z^{\top}z \leq x} \left[\phi(z) + \sum_{i = 1}^{2} n^{-i/2} W_{i}(z)\right]  dz + O(n^{-5/6}).
\end{equation*}	
	
Note that $ W_{1} $ is an odd function, thus it integrates to 0 over the region $ z^{\top}z \leq x $. The only term left is $ W_{2} $. But for $ i \neq j, W_{2} $ is also odd over the region $ z^{\top}z \leq x $, thus we only need to consider $ W_{2}(x) = \frac{1}{2} (a_{ii} - 2q \frac{n}{Q} a)(x^i x^i - q) \phi(x) $. That is
\begin{align*}
	\int_{z^{\top}z \leq x} W_{2}(z) dz & = \int_{z^{\top}z \leq x}\frac{1}{2} \left(a_{ij} -  2 a \delta_{ij} \frac{n}{Q} \right) (z^{i} z^{j} - \eta_{ij}) \phi(x) dz \\
	& = \frac{1}{2} \left(a_{ii} -  2 a q \frac{n}{Q} \right) \int_{z^{\top}z \leq x}(z^{i} z^{i} - q) \phi(z) dz.
\end{align*}
	
Therefore, by letting
\begin{equation*}
    a = \frac{1}{2} \frac{Q}{n} \frac{1}{q} a_{ii},
\end{equation*}
the $ n^{-1} W_{2}$ term vanishes. As a result,
\begin{equation*}
    P(nR^{\top}R) \leq x) = P(\chi^{2}_{q} \leq x) + O(n^{-5/6}).
\end{equation*}	
This completes the proof.
\end{proof}

\section*{Supplementary Materials}

Detailed calculations in the proofs and the complete simulation results are available in the online supplement. R code used for the simulation given in Section \ref{sec: adj simulation} and data application in Section \ref{sec: applicatioin} are available on github \url{https://github.com/kwgit/ABEL_companion_code}.

\par

\par


\bibhang=1.7pc
\bibsep=2pt
\fontsize{9}{14pt plus.8pt minus .6pt}\selectfont
\renewcommand\bibname{\large \bf References}
\expandafter\ifx\csname
natexlab\endcsname\relax\def\natexlab#1{#1}\fi
\expandafter\ifx\csname url\endcsname\relax
  \def\url#1{\texttt{#1}}\fi
\expandafter\ifx\csname urlprefix\endcsname\relax\def\urlprefix{URL}\fi

\bibliographystyle{chicago}      
\bibliography{ABEL.bib}   

\end{document}